\documentclass[11pt,letterpaper]{article}

	\title{Volatility of Volatility and Leverage Effect from Options\footnote{We would like to thank the Editor, an Associate Editor, anonymous referees as well as participants at various conferences for many useful comments and suggestions.}}

	\author{Carsten H.\ Chong\thanks{Department of Information Systems, Business Statistics and Operations Management,
			The Hong Kong University of Science and Technology, e-mail: carstenchong@ust.hk}	\and Viktor Todorov\thanks{Department of Finance, Northwestern University, e-mail: v-todorov@kellogg.northwestern.edu}}
	\date{}

\usepackage{amsmath, amsfonts, amsthm, verbatim,
	amssymb,dsfont,color, accents}

\usepackage[utf8x]{inputenc}
\usepackage{mathtools}

\usepackage{tikz}
\usetikzlibrary{arrows.meta,arrows,decorations.pathreplacing}
\RequirePackage[colorlinks,citecolor=blue,urlcolor=blue]{hyperref}
\usepackage{tabularx,booktabs}
\usepackage[round]{natbib}
\usepackage{multirow}

\newcommand{\R}{\mathbb{R}}

\newcommand{\N}{\mathbb{N}}
\newcommand{\E}{\mathbb{E}}
\newcommand{\F}{\mathbb{F}}

\renewcommand{\P}{\mathbb{P}}
\newcommand{\Q}{\mathbb{Q}}

\newcommand{\bone}{\mathbf 1}

\theoremstyle{plain}
\newtheorem{theorem}{Theorem}[section]
\newtheorem{lemma}[theorem]{Lemma}
\newtheorem{proposition}[theorem]{Proposition}
\newtheorem{corollary}[theorem]{Corollary}
\newtheorem{condition}{Condition}
\newtheorem{Assumption}{Assumption}
\renewcommand{\theAssumption}{\Alph{Assumption}}
\makeatletter
\newcommand{\settheoremtag}[1]{
	\let\oldtheAssumption\theAssumption
	\renewcommand{\theAssumption}{#1}
	\g@addto@macro\endAssumption{
		\global\let\theAssumption\oldtheAssumption}
}
\makeatother
\theoremstyle{remark}
\newtheorem{remark}[theorem]{Remark}

\newcommand{\calu}{{\cal U}}
\newcommand{\calf}{{\cal F}}

\newcommand{\call}{{\cal L}}

\newcommand{\calg}{{\cal G}}

\newcommand{\calw}{{\mathcal{W}}}
\newcommand{\pf}{{\frak p}}
\newcommand{\qf}{{\frak q}}

\newcommand{\al}{{\alpha}}
\newcommand{\la}{{\lambda}}

\newcommand{\eps}{{\epsilon}}
\newcommand{\ga}{{\gamma}}

\newcommand{\vp}{{\varphi}}
\newcommand{\si}{{\sigma}}
\newcommand{\Si}{{\Sigma}}

\newcommand{\Om}{{\Omega}}

\newcommand{\ov}{\overline}
\newcommand{\un}{\underline}
\newcommand{\wh}{\widehat}
\newcommand{\wt}{\widetilde}

\newcommand{\bj}{{\boldsymbol{j}}}

\newcommand{\bz}{{\boldsymbol{z}}}

\newcommand{\uc}{\mathrm{uc}}

\newcommand{\Den}{\Delta_n}
\newcommand{\avar}{\mathrm{AVar}}

\newcommand{\bthm}{\begin{theorem}}
	\newcommand{\ethm}{\end{theorem}}

\newcommand{\bcor}{\begin{corollary}}
	\newcommand{\ecor}{\end{corollary}}

\newcommand{\blem}{\begin{lemma}}
	\newcommand{\elem}{\end{lemma}}

\newcommand{\bprop}{\begin{proposition}}
	\newcommand{\eprop}{\end{proposition}}

\newcommand{\bcond}{\begin{condition}}
	\newcommand{\econd}{\end{condition}}

\newcommand{\bdf}{\begin{definition}}
	\newcommand{\edf}{\end{definition}}

\newcommand{\bex}{\begin{example}}
	\newcommand{\eex}{\end{example}}

\newcommand{\brem}{\begin{remark}}
	\newcommand{\erem}{\end{remark}}

\newcommand{\bpr}{\begin{proof}}
	\newcommand{\epr}{\end{proof}}

\newcommand{\benu}{\begin{enumerate}}
	\newcommand{\eenu}{\end{enumerate}}

\newcommand{\beq}{\begin{equation}}
	\newcommand{\eeq}{\end{equation}}

\newcommand{\bit}{\begin{itemize}}
	\newcommand{\eit}{\end{itemize}}

\newcommand{\bass}{\begin{Assumption}}
	\newcommand{\eass}{\end{Assumption}}

\numberwithin{equation}{section}

\allowdisplaybreaks[4]

\begin{document}

\maketitle

\begin{abstract}
\noindent We propose model-free (nonparametric) estimators of the volatility of volatility and leverage effect using high-frequency observations of short-dated options. At each point in time, we integrate available options into estimates of the conditional characteristic function of the price increment until the options' expiration and we use these estimates to recover spot volatility. Our volatility of volatility estimator is then formed from the sample variance and first-order autocovariance of the spot volatility increments, with the latter correcting for the bias in the former due to   option observation errors. The leverage effect estimator is the sample covariance between price increments and the estimated volatility increments. The rate of convergence of the estimators depends on the diffusive innovations in the latent volatility process as well as on the observation error in the options with strikes in the vicinity of the current spot price. Feasible inference is developed in a way that does not require prior knowledge of the source of estimation error that is asymptotically dominating.     
\end{abstract}

\bigskip

\noindent
{\em JEL Classification:}  C14, C22, C58, G12.

\bigskip

\noindent
{\em Keywords:}
Characteristic function; leverage effect; nonparametric estimation; options; volatility of volatility.

\section{Introduction}
	
Time-varying volatility is a ubiquitous feature of asset prices. Given the importance of volatility risk for investors, there is a lot of interest in trading derivative products written on the market volatility index VIX. Volatility of volatility is a major source of risk for traders of volatility derivatives. Related to that, there is a growing literature in finance that documents the emergence of volatility of volatility as a separate risk factor and its weak correlation with volatility risk, see e.g., \cite{agarwal2017volatility}, \cite{hollstein2018aggregate}, \cite{huang2019volatility} and \cite{chen2022volatility} among others. 

Another stylized feature of volatility is its tendency to move in the opposite direction to the asset price. Following \cite{black1976studies}, this negative covariance/correlation between asset return and asset volatility is referred to as leverage effect. The economic origins of the leverage effect have been extensively studied in earlier work in finance, see e.g., \cite{black1976studies}, \cite{christie1982stochastic}, \cite{french1987expected}, \cite{campbell1992no}, \cite{engle1993measuring}, \cite{bekaert2000asymmetric}, \cite{bollerslev2012volatility} and \cite{bollerslev2023jump}, among others.\footnote{Alternative discrete-time stochastic volatility models with leverage effect have been estimated by \cite{yu2005leverage}.}  

Since stochastic volatility is not directly observed, estimating the spot volatility of volatility and the spot leverage effect from asset returns is rather challenging. In this paper, we propose nonparametric estimators for these quantities using options written on the underlying asset. The proposed estimators can be used as diagnostic tools for stochastic volatility modeling. They can be also used for studying the asset pricing implications of volatility of volatility risk and for separating alternative economic explanations regarding the role of volatility risk in equilibrium asset pricing models as done in some of the work mentioned in the previous paragraph. 

One way to estimate the volatility of volatility and the leverage effect is to use high-frequency asset return data. More specifically, one can first construct local estimators of volatility and then form realized variance of variance and realized covariance between   price and   variance from these estimates, see e.g., \cite{vetter2015estimation}, \cite{sanfelici2015high}, \cite{clinet2021estimation}, \cite{li2022volatility} and \cite{toscano2022volatility} for volatility of volatility and \cite{WM14}, \cite{AFLWY17}, \cite{KX17} and \cite{yang2023estimation} for the leverage effect. The resulting estimators of integrated volatility of volatility and leverage effect, however, have a relatively slow rate of convergence. This reflects the difficulty of such an estimation problem. More specifically, if the asset prices do not contain market microstructure noise, the best attainable rate for integrated volatility of volatility and leverage effect is $n^{1/4}$, where $n$ is the number of high-frequency increments in the fixed time interval. If asset prices are contaminated with noise, this rate drops to $n^{1/8}$. Furthermore, if one is interested in spot instead of integrated volatility  of volatility and leverage effect, then the above rates drop  by another factor of one half. In addition, the results for the return-based volatility of volatility estimators are derived under various nontrivial restrictions for the underlying asset price, e.g., \cite{vetter2015estimation} does not allow for price and volatility jumps while \cite{li2022volatility} rule out volatility jumps.
	
This paper  proposes an alternative approach for volatility of volatility and leverage effect estimation that is based on short-dated options, i.e., options with very short time-to-maturity.\footnote{Options written on different underlying assets and expiring within a few weeks are actively traded on exchanges. This is facilitated by option exchanges offering the so-called weekly options, i.e., options that expire on a weekly basis, with the shortest ones expiring in a week from the issue date. Starting in 2022, the CBOE options exchange offers options on the S\&P 500 index that expire at the end of each trading day. In 2022, more than 60\% of the trading volume of  S\&P 500 index options was in options expiring within 7 calendar days.} 
The idea is to replace the local estimates of spot volatility from high-frequency returns with ones formed from short-dated options. As is well known, short-dated options can render the latent volatility directly observable and that makes the problem of nonparametric estimation of volatility of volatility and the leverage effect significantly easier.\footnote{For recovery of the spot volatility in a model-free way, we need the time-to-maturity of the options to be shrinking asymptotically.} Indeed, \cite{todorov2022information} document nontrivial gains from using option data for spot volatility estimation. This advantage naturally caries over to the case of estimating volatility of volatility and the leverage effect.\footnote{\cite{BFR23,BFR23b} also use short-dated options data to estimate the spot volatility of volatility and  the spot leverage effect (among other asset price characteristics), without  a proof of consistency or asymptotic normality of these estimators. This is done in a semiparametric setting in which price jumps are modeled parametrically and there are no volatility jumps. Our procedure, by contrast, is fully nonparametric and allows for volatility jumps (either positive or negative) in particular.}
	
In our inference procedures we use the spot volatility estimator of \cite{T19}, which is formed from option-based nonparametric estimates of the conditional characteristic function of   price increments. \cite{T19} shows that, for any fixed $u>0$, $\mathbb{E}_t[e^{iu(x_{t+T}-x_t)/\sqrt{T}}]$ can be used to recover nonparametrically the spot diffusive volatility of a general It\^o semimartingale process $x$ when $T$ is asymptotically shrinking.\footnote{In our application, we use $T$ of up to $16$ business days and the median highest $T$ on each day in our sample is only $6$ business days.} As usual, $\mathbb{E}_t$ in the above denotes conditional expectation given information up to time $t$. In \cite{CT23_a}, we derive a higher order expansion for $\mathbb{E}_t[e^{iu(x_{t+T}-x_t)/\sqrt{T}}]- \mathbb{E}_{t-\Delta}[e^{iu(x_{t+T}-x_{t-\Delta})/\sqrt{T+\Delta}}]$ when both $\Delta$ and $T$, with $0<\Delta<T$, are asymptotically shrinking. We use this result here to quantify the difference between the high-frequency increments of our option-based volatility estimator and the infeasible high-frequency increments of the latent spot diffusive volatility. This, in turn, can be used to estimate nonparametrically various quantities associated with the volatility dynamics, and in particular the volatility of volatility and the leverage effect. 

Our volatility of volatility estimator is based on sample variance and first-order autocovariance of the option-based spot volatility increments. The reason for using autocovariance in the estimation is to account for the effect of option observation errors. Option observation errors can be viewed as the natural counterpart to the market microstructure noise in the underlying asset price. Their presence introduces an upward bias in the sample variance of the volatility returns. The first-order autocovariance of the variance return estimates can correct for it. Turning next to the leverage effect, its estimator is constructed in an analogous way from the sample covariance between the price and the estimated volatility increments. 
	
We show consistency of the two estimators and derive a Central Limit Theorem (CLT) for them. The limit distribution is determined by two sources of error in the estimation procedure.  One is due to the diffusive innovations in the true (latent) volatility process. The other is due to the observation errors in the options with strikes in the vicinity of the current stock price. The size of the two errors is governed by the time gap between observations for the first one and the tenor and the strike gaps for the second one. All of these quantities are asymptotically shrinking and for our estimation procedure we do not require a condition on their relative size. 

The performance of the proposed estimators is evaluated on simulated data from an asset pricing model calibrated to match key features of observed stock and option data. Our Monte Carlo shows that the option-based estimators of volatility of volatility and the leverage effect  are significantly more efficient than their return-based counterparts. This is the case in spite of the fact that the latter are constructed using a much higher sampling frequency than the option based ones (five seconds versus one minute). In an empirical application, we compute the volatility of volatility and the leverage effect of the S\&P 500  market index using short-dated options written on the index. We use a log transform of the volatility when computing these quantities, i.e., our interest is volatility of log-volatility and covariance between log-price and log-volatility. Our nonparametric results show that volatility and volatility of volatility exhibit only weak dependence while the leverage effect estimates and the market variance are inversely related. 

On a theoretical level, the papers closest to the current work are \cite{andersen2015exploring} and \cite{KX17}. These papers propose volatility of volatility and leverage effect estimators using the VIX volatility index instead of the nonparametric spot volatility estimators employed here. There are two major differences between the estimators of these papers and ours. First, the VIX index is an estimate of the conditional risk-neutral expectation of one month ahead return quadratic variation. It has dynamics which is different from that of the latent spot diffusive volatility. Hence, the estimands of \cite{andersen2015exploring} and \cite{KX17} on one hand and of our paper on the other hand are different. Second, \cite{andersen2015exploring} and \cite{KX17} do not allow for option measurement error while we do. In fact, both the construction of our estimators and their limiting distributions are impacted by the presence of option observation errors.  
	
The rest of the paper is organized as follows. In Section~\ref{sec:setting}, we introduce our setting and define the objects of interest in the paper. Section~\ref{sec:options} describes the option observation scheme and the construction of volatility estimators from options. Our estimators and the asymptotic results for them are given in Section~\ref{sec:theory}.  Section~\ref{sec:mc} contains Monte Carlo evidence and Section~\ref{sec:emp} our empirical application.  Section~\ref{sec:concl} concludes. Some asymptotic expansion formulas for volatility estimators based on characteristic functions are stated in Appendix~\ref{sec:exp}, while the proofs of the theoretical results of this paper are given in Appendix~\ref{sec:proofs}.          
	
\section{Setting}\label{sec:setting}

The logarithm of the asset price is denoted by $x$ and is defined  on a filtered  space $(\Om,\calf,\F=(\calf_t)_{t\geq0})$, equipped with two probability measures, the true statistical probability measure $\P$ and the risk-neutral probability measure $\Q$. We assume that $x$ is an It\^o semimartingale under $\Q$ of the form
\begin{equation}\label{eq:x0}  
\begin{split}
x_t&=x_0+\int_0^t\al_s ds +\int_0^t \si_s dW_s +J^x_t, 
\end{split}
\end{equation}
where $W$ is a Brownian motion, $\al$ and $\si$ denote the drift and volatility of the asset price, respectively, and $J^x$ is the jump part of $x$ (which is given in \eqref{eq:jumps}). 

In this paper, we are interested in the volatility of the volatility process $\si$ and its components. Assuming that
\begin{equation}\label{eq:si0} \begin{split}
	\si_t&=\si_0+\int_0^t\al^\si_s ds + \int_0^t \si^\si_s dW_s + \int_0^t \ov\si^\si_s d\ov W_s +J^\si_t,
\end{split} 
\end{equation}
we develop feasible inference for $(\si^\si_t)^2+(\ov\si^\si_t)^2$ and $\si_t\si^\si_t$. The first of these two quantities is the spot diffusive variance of the process $\sigma$, while $\si_t\si^\si_t$ captures the covariance between diffusive price and volatility moves. In \eqref{eq:si0}, $\al^\si$, $\si^\si$ and $\ov\si^\si$ are the drift and diffusive coefficients of $\si$, respectively,  $\ov W$ is a Brownian motion independent of $W$, and $J^\si$ denotes the jump component of $\si$ (given in \eqref{eq:jumps}).

Our estimation procedure is based on options written on the underlying asset whose theoretical values are related to the dynamics of $x$ under the risk-neutral probability measure $\Q$, see equation (\ref{eq:opt}) below. In an arbitrage-free setting, the latter is locally equivalent to the statistical probability measure $\P$, and by Girsanov's theorem (see Theorem~III.3.24 in \cite{JS03}), both $x$ and $\si$ are also  It\^o semimartingales under $\P$, with representations of the form \eqref{eq:x0} and \eqref{eq:si0}, respectively. To be notationally precise, one should put superscripts $\P$ or $\Q$ on the drift, the jump part and the Brownian motions in \eqref{eq:x0} and \eqref{eq:si0}. In order to simplify notation, however, we shall refrain from doing so and instead mention the probability measure whenever we refer to \eqref{eq:x0} or \eqref{eq:si0}. Let us also stress at this point that the main quantities we are interested in here, namely, $\si$,  $\si^\si$ and $\ov\si^\si$, are the same under $\P$ and $\Q$.

As mentioned in the introduction and further explained in Section~\ref{sec:options} below, prices at time $t$ for options with time-to-maturity $T$ allow us to reconstruct the $\calf_t$-conditional characteristic function (under $\Q$) of the normalized price change from $t$ to $t+T$, that is, of
\begin{equation}\label{eq:call} 
\call_{t,T}(u)= \E_t[e^{iu(x_{t+T}-x_t)/\sqrt{T}}] = \E_t[e^{iu_T(x_{t+T}-x_t)}],
\end{equation} 
where $u_T=u/\sqrt{T}$ and $\E_t=\E^\Q[\cdot \mid \calf_t]$ denotes $\calf_t$-conditional expectation under $\Q$. Using the fact that $\call_{t,T}(u)=e^{-\frac12 u^2\si_t^2} + o(1)$ as $T\downarrow 0$ (under mild assumptions on $x$ and $\si$), \cite{T19} constructs an estimator of the spot variance $\si_t^2$   by setting
\begin{equation}\label{eq:vol-est} 
	\si^2_{t,T}(u)=-\frac{2}{u^2} \log \lvert \call_{t,T}(u)\rvert. 
\end{equation} 
Here, and in the remainder of the paper, we   use $o=o_p$ and $O=O_p$ to indicate order in probability. We further add the  superscript ``uc'' as in $O^\uc$ or $o^\uc$ to indicate uniformity in $u\in\calu$, where $\calu$ is an arbitrary compact subset of $(0,\infty)$. In order to remove biases of higher asymptotic order of this estimator,  \cite{todorov2021bias} introduce a bias-corrected version of \eqref{eq:vol-est} by considering a second time-to-maturity
 \begin{equation}\label{eq:Tprime} 
 	T'=\tau T,
 \end{equation}
 for some $\tau>1$ and defining
\begin{equation}\label{eq:vol-est-20} 
 \si^2_{t,T,T'}(u)=\frac{T'\si^2_{t,T}(u)-T\si^2_{t,T'}(u)}{T'-T}.
\end{equation}
The main sources of error of both $\si^2_{t,T}(u)$ and $\si^2_{t,T,T'}(u)$ are jump risks in price and volatility and the dynamics of their semimartingale characteristics. 

Our strategy of estimating spot volatility of volatility consists of forming sample variance and autocovariance of high-frequency increments of $\si^2_{t,T}(u)$ and $\si^2_{t,T,T'}(u)$, that is, of
\begin{equation}\label{eq:si-incr} \begin{split}
\Delta^n_i \si^2_{t,T}(u)	&=\si^2_{t^n_{i-1},T^n_{i-1}}(u)-\si^2_{t^n_i,T^n_i}(u),\\
\Delta^n_i \si^2_{t,T,T'}(u)&=\si^2_{t^n_{i-1},T^n_{i-1}, T^{\prime n}_{i-1}}(u)-\si^2_{t^n_{i},T^n_{i}, T^{\prime n}_{i}}(u),
\end{split}
\end{equation}
where with a slight abuse of notation we let
\begin{equation}\label{eq:tT} 
	t^n_i=t-i\Den,\quad T^n_i = T+i\Den,\quad T^{\prime n}_i = T'+i\Den,\quad i=1,\dots,k_n,
\end{equation} 
for some $\Den\to0$ and $k_n\to\infty$ with $k_n\Den\to0$. In \cite{CT23_a}, we have derived higher-order asymptotic expansions for $\Delta^n_i \si^2_{t,T}(u)$ and $\Delta^n_i \si^2_{t,T,T'}(u)$, as $\Den\to0$ and $T\to 0$. We will make advantage of these results, which we recall in  Appendix~\ref{sec:exp}, to develop our volatility of volatility and leverage effect estimators.
	
In practice, one is often interested in estimating volatility or variance of a \emph{transform} of volatility, that is, of $V_t=F(\si^2_t)$, where $F$ is a $C^2$-function on $(0,\infty)$. Typical functions of interest include $F(x)=x$ (volatility of variance), $F(x)=\sqrt{x}$ (volatility of volatility), $F(x)=\log x$ (volatility of log-variance) and $F(x)= \log \sqrt{x}$ (volatility of log-volatility). For any fixed $0\leq\underline t < t<\overline t<\infty$, since $\si$ is an It\^o semimartingale, $V$ is again an It\^o semimartingale on $[\underline t, \overline t]$ on the event $\{\inf_{s\in[\underline t,\overline t]}\si_s^2>0\}$. We are interested in estimating $VV_t$, the spot variance of $V$ at time $t$ and $LV_t$, the spot covariance of $V$ and the asset price. By It\^o's formula, we have
\begin{equation}\label{eq:VV} 
VV_t= 4\si_t^2( F'(\si_t^2))^2 ((\si^\si_t)^2+(\ov \si^\si_t)^2)~~\textrm{and}~~LV_t = 2\sigma_t^2F'(\si_t^2)\si^\si_t,
\end{equation}
which are the same under $\P$ and $\Q$. As natural but infeasible estimators of $V_t$, we consider
\begin{equation}\label{eq:V} 
V_{t,T}(u) = F(\si^2_{t,T}(u)),\qquad 	V_{t,T,T'}(u)=\frac{T'V_{t,T}(u)-TV_{t,T'}(u)}{T'-T},
\end{equation}
where $\si^2_{t,T}(u)$ is defined in \eqref{eq:vol-est}. Similarly, infeasible approximations of $\Delta^n_i V_t$ are given by
\begin{equation}\label{eq:incr-V} \begin{split}
		\Delta^n_i V_{t,T}(u)	&=V_{t^n_{i-1},T^n_{i-1}}(u)-V_{t^n_i,T^n_i}(u),\\
		\Delta^n_i V_{t,T,T'}(u)&=V_{t^n_{i-1},T^n_{i-1}, T^{\prime n}_{i-1}}(u)-V_{t^n_{i},T^n_{i}, T^{\prime n}_{i}}(u).
	\end{split}
\end{equation}

\section{Option-Based Volatility Estimators}\label{sec:options}
	
We proceed next with   constructing estimators of $V_{t,T}(u)$ and $V_{t,T,T'}(u)$ from options. If $x$ denotes the log-price of an asset, the conditional characteristic function of its  increments can be inferred from portfolios of short-dated options using an option-spanning result, see \cite{bakshi2000spanning} and \cite{CM01}. More specifically, assuming that the dividend yield associated with $x$ and the risk-free interest rate are both equal to zero (as their effect on short-dated options is negligible), we have
\begin{equation}\label{eq:spanning}
\call_{t,T}(u) = 1-\left(\frac{u^2}{T}+i\frac{u}{\sqrt{T}}\right)e^{-x_t}\int_{\mathbb{R}}e^{(iu/\sqrt{T}-1)(k-x_t)}O_{t,T}(k)dk,
\end{equation}
for $u\in\R$, where $O_{t,T}(k)$ denotes the price at time $t$ of an European style out-of-the-money option expiring at $t+T$ and with log-strike of $k$, that is, 
\begin{equation}\label{eq:opt} 
O_{t,T}(k)=\begin{cases} \E_t[(e^k- e^{x_{t+T}})\vee 0] &\text{if } k\leq x_t, \\ \E_t[( e^{x_{t+T}}-e^k)\vee 0] &\text{if } k> x_t. \end{cases}
\end{equation}
In agreement with the notation used so far, $\E_t$ signifies the $\calf_t$-conditional expectation under $\mathbb{Q}$.
We remind the reader that $O_{t,T}(k)$ is a put if $k\leq \log(F_{t,T})$ and a call otherwise, where $F_{t,T}$ is the time-$t$ futures price of the asset with expiration date $t+T$.
	
If the option prices $O_{t,T}(k)$ were continuously observable in $k$, then thanks to \eqref{eq:spanning}, the conditional characteristic function $\call_{t,T}(u)$, and hence $\si^2_{t,T}(u)$, would be statistics that one could  use to estimate $\si^2_t$, for example. In practice, there are two complications. First, $O_{t,T}(k)$ is only available on a  discrete log-strike grid, say, for
\begin{equation}\label{eq:logmoney}
\underline{k}_{t,T} ~\equiv~ k_{1,t,T} \, < \, k_{2,t,T} \, < \, \cdots \, < \, k_{N_{t,T},t,T}~\equiv~ \overline{k}_{t,T},\qquad N_{t,T}\in\mathbb{N}_+,
\end{equation}
which may be random and vary in $t$ and $T$. We denote the gap between consecutive log-strikes  by $\delta_{j,t,T} = k_{j,t,T} - k_{j-1,t,T}$, for $j=2,\dots,N_{t,T}$. Second, the observed option prices contain errors. That is, we only observe
\begin{equation}\label{eq:obs}
\widehat{O}_{t,T}(k_{j,t,T}) = O_{t,T}(k_{j,t,T})+\epsilon_{j,t,T},\qquad j=1,\dots,N.
\end{equation}
We assume that the errors $\epsilon_{j,t,T}$ are defined on an auxiliary space $(\Omega^{(1)},\mathcal{F}^{(1)})$ equipped with a transition probability $\mathbb{P}^{(1)}(\omega,d\omega^{(1)})$ from $\Omega$, the probability space on which $x$ is defined, to $\Omega^{(1)}$. We further define
\begin{equation}\label{eq:ext} 
\ov\Omega \,=\, \Omega\times \Omega^{(1)},\quad\ov{\mathcal{F}} \,=\, \mathcal{F} \otimes \mathcal{F}^{(1)},\quad\ov{\mathbb{P}}(d\omega,d\omega^{(1)}) \,=~ \mathbb{P}(d\omega) \, \mathbb{P}^{(1)}(\omega,d\omega^{(1)}).
\end{equation}   
Making a simple Riemann sum approximation of the integral in (\ref{eq:spanning}) using the available options, a feasible estimator of $\call_{t,T}(u)$ is now given by 
\begin{equation}\label{eq:L_hat}
\widehat{\mathcal{L}}_{t,T}(u) = 1 - \left(\frac{u^2}{T}+i\frac{u}{\sqrt{T}}\right)e^{-x_t}\sum_{j=2}^{N_{t,T}}e^{(iu/\sqrt{T}-1)(k_{j-1,t,T}-x_t)}\widehat{O}_{t,T}(k_{j-1,t,T})\delta_{j,t,T},
\end{equation}
for $u\in\R$. This in turn leads to feasible versions of the estimators from \eqref{eq:V} via
\begin{equation}\label{eq:V_hat} 
\wh V_{t,T}(u)=F(\widehat{\si}^2_{t,T}(u)),	\qquad \widehat{\si}^2_{t,T}(u) = -\frac{2}{u^2}\log|\widehat{\mathcal{L}}_{t,T}(u)|,
\end{equation}
and
\begin{equation}\label{eq:V_12}
\wh V_{t,T,T'}(u)= \frac{T'\widehat{V}_{t,T}(u) - T\widehat{V}_{t,T'}(u)}{T'-T}.
\end{equation}
	
For $\wh\call_{t,T}(u)$ to be a sufficiently good approximation of $\call_{t,T}(u)$, we need several assumptions concerning the existence of conditional moments of $x$ under $\Q$, the option observation scheme as well as the observation errors. They are quite similar to those employed in \cite{T19} and \cite{todorov2021bias}.   In the following, if expectation is taken under $\mathbb{Q}$, we will not use superscript in the notation; if expectation is under $\P$ or $\ov \P$, we put superscripts to signify this.
	
\bass\label{ass:C} The observed option prices are defined on $(\ov\Om,\ov\calf,\ov\P)$ from \eqref{eq:ext} and satisfy \eqref{eq:obs}. Moreover, there exists an $\F$-adapted process  $C_t$  with c\`{a}dl\`{a}g paths such that the following holds:
	\benu
	\item  For all $0<t<u<\infty$, we have that
		\begin{equation}\label{a3:1}
	\E_t\biggl[\al_u^4+\sigma_u^6+e^{4\lvert x_u\rvert}+\biggl(\int_\R \bigl[(e^{3\lvert z\rvert}-1)\vee \lvert z\rvert^2\bigr] \la(dz)\biggr)^4\biggr]<C_t.
	\end{equation}
	\item The number of strikes $N_{t,T}$ and the log-strike grid $\{k_{j,t,T}\}_{j=1}^{N_{t,T}}$ are $\mathcal{F}_{{t}}$-measurable  and  
	\begin{equation}\label{a4:1}
		C^{-1}_{{t}}\delta\leq \delta_{j,t,T}\leq C_{{t}}\delta,\qquad j=2,\dots, N_{t,T},
	\end{equation}
	for a deterministic sequence $\delta=\delta(T)$. Moreover, for any $\tau>0$,
	\[ \sup_{j:\lvert k_{j,t,\tau T}-x_t\rvert < C_t^{-1}}  \lvert \delta_{j,t,\tau T}/\delta-\rho_{t,\tau}(k_{j-1,t,\tau T}-x_t) \rvert \stackrel{\P}{\longrightarrow} 0, \]
	where $t\mapsto \rho_{t,\tau}(k)$ is continuous in probability (uniformly in $k$), $\F$-adapted, continuous in $\tau>0$ (uniformly in $k$ and locally uniformly in $t$) and continuous in $k=0$ (locally uniformly in $\tau$ and $t$).
	\item In the notation of \eqref{eq:tT} and \eqref{eq:Tprime},
	\[ \liminf_{T\to0} \frac{\inf_{n\in\N, i=1,\dots,k_n} ( \lvert\underline  k_{t^n_i,T^n_i} \rvert\wedge\lvert \underline k_{t^n_i,T^{\prime n}_i}\rvert\wedge \overline  k_{t^n_i,T^n_i}  \wedge  \overline k_{t^n_i,T^{\prime n}_i}  )}{(\delta/\sqrt{T})^\iota} = \infty  \]
	for some $\iota>0$.
	\item For $t,\tau>0$ and $j=1,\dots, N_{t,\tau T}$, we have
	\begin{equation}\label{eq:epsilon} 
		\epsilon_{t,\tau T}(k_{j,t,\tau T}) = \zeta_{t,\tau}(k_{j,t,\tau T}-x_t)O_{t,\tau T}(k_{j,t,\tau T})\overline{\epsilon}_{j,t,\tau T},
	\end{equation}
	where $t\mapsto \zeta_{t,\tau}(k)$ is continuous in probability (uniformly in $k$), $\F$-adapted, continuous in $\tau>0$ (uniformly in $k$ and locally uniformly in $t$) and continuous in $k=0$ (locally uniformly in $\tau$ and $t$) and $\overline{\epsilon}_{j,t,T}$ is   $\ov\calf$-measurable, independent of $\calf$ under $\ov\P$ and i.i.d.\ as $j$, $t$ and $T$ vary. Moreover,
	\begin{equation}\label{eq:mom} 
		\mathbb{E}^{\ov \P}[\overline{\epsilon}_{j,t,T}\mid \mathcal{F}]= 0,\quad \mathbb{E}^{\ov \P}[(\overline{\epsilon}_{j,t,T})^2\mid \mathcal{F}] = 1,\quad \mathbb{E}^{\ov \P}[| \overline{\epsilon}_{j,t,T}|^{p}\mid\mathcal{F}] <\infty\quad \text{for all } p>2.
	\end{equation} 
	\eenu
	\eass

\section{Volatility of Volatility and Leverage Effect Estimators}\label{sec:theory}

We are now ready to introduce our estimators. Under Assumption~\ref{ass:C}, we can use the observed option prices to build reasonable estimators of $\call_{t^n_i,T^n_i}(u)$ and $\call_{t^n_i,T^{\prime n}_i}(u)$; see \cite{T19}. If we further have Assumptions \ref{ass:main} and \ref{ass:main-1} stated in the Appendix, Theorem~\ref{thm:main0} shows that $\Delta^n_i V_{t,T}(u)$ and $\Delta^n_i V_{t,T,T'}(u)$ provide good approximations of $\Delta^n_i V_t$. These two results combined can be used to construct estimators of $VV_t$ and $LV_t$. 

We start with  volatility of volatility. Our estimators of $VV_t$ are based on the sample variance and first-order autocovariance of these approximations of $\Delta^n_i V_t$. More precisely, given a truncation function $\tau_n(x)=x\bone_{\{\lvert x\rvert\leq \upsilon_n\}}$ where $\upsilon_n>0$ is some truncation level, we consider for $u>0$ the following volatility of volatility estimators:
\begin{equation}\label{eq:vov}\begin{split}
\wh{VV}^n_{t,T}(u)&=\frac{1}{k_n\Den} \sum_{i=2}^{k_n}\Bigl((\tau_n(\Delta_i^n\widehat{V}_{t,T}(u)))^2 + 2\tau_n(\Delta_{i-1}^n\widehat{V}_{t,T}(u))\tau_n(\Delta_i^n\widehat{V}_{t,T}(u))\Bigr), \\
\wh{VV}^n_{t,T,T'}(u)&=\frac{1}{k_n\Den} \sum_{i=2}^{k_n}\Bigl((\tau_n(\Delta_i^n\widehat{V}_{t,T,T'}(u)))^2 + 2\tau_n(\Delta_{i-1}^n\widehat{V}_{t,T,T'}(u))\tau_n(\Delta_i^n\widehat{V}_{t,T,T'}(u))\Bigr),
\end{split}\end{equation}
where, recalling the notation in \eqref{eq:tT}, \eqref{eq:V_hat} and \eqref{eq:V_12}, we define
	\[ \Delta_i^n\widehat{V}_{t,T}(u)=\wh V_{t^n_{i-1},T^n_{i-1}}(u)-\wh V_{t^n_i,T^n_i}(u),\quad \Delta_i^n\widehat{V}_{t,T,T'}(u)=\wh V_{t^n_{i-1},T^n_{i-1},T^{\prime n}_{i-1}}(u)-\wh V_{t^n_i,T^n_i,T^{\prime n}_i}(u).\]
The reason we include the first-order autocovariance for the volatility of volatility estimators is because this automatically corrects for a bias caused by the option observation errors. 

Turning next to the leverage effect, our estimators  are based on sample covariances involving the price increments and the increments of the volatility estimates:
\begin{equation}\label{eq:lev}\begin{split}
\wh{LV}^n_{t,T}(u)& = \frac{1}{k_n\Den} \sum_{i=1}^{k_n}\tau_n(\Delta_i^nx_t)\tau_n(\Delta_i^n\widehat{V}_{t,T}(u)),\\
\wh{LV}^n_{t,T,T'}(u) &= \frac{1}{k_n\Den} \sum_{i=1}^{k_n}\tau_n(\Delta_i^nx_t)\tau_n(\Delta_i^n\widehat{V}_{t,T,T'}(u)).\end{split}
\end{equation}
Note that here, unlike the case of estimating $VV_t$, we do not need to make corrections for the option observation error as the latter does not introduce an asymptotic bias in $\wh{LV}^n_{t,T}(u)$ and $\wh{LV}^n_{t,T,T'}(u)$.

In order to state the theoretical results of this paper, we need one more set of assumptions, which relate $\Delta_n$, $T$, $k_n$, the strike grid size $\delta$ and the truncation threshold $\upsilon_n$ to each other.
\begin{Assumption}\label{ass:hf} ~
	We have $\Den\to0$, $T\to 0$,  $k_n\to \infty$, $\delta\to0$ and $\upsilon_n >0$ in such a way that for some $N\geq3$, some $\phi\in[0,1]$ and some $\iota>0$,
		\begin{equation}\label{eq:rel} \begin{split}
				&\frac{k_nT^{N}}{\Den} \to 0,\quad k_nT=O(1),\quad  k_n\Den=O(T),\quad k_n^2\Den\to0,\quad \frac{k_n\delta\log T}{\sqrt{T}}\to0,\\
				&\frac{\Den}{\Den+\delta/\sqrt{T}}\to \phi,\quad	\Den=O(\upsilon_n^{2r}),\quad \frac{\Den+\delta/\sqrt{T}}{\upsilon_n^2} = O(k_n^{-\iota}).
			\end{split}
		\end{equation} 
	Furthermore, for the convergence of $\wh{VV}^n_{t,T}(u)$ in \eqref{eq:main-conv-alt} as well as the convergence of  
	$\wh{LV}^n_{t,T}(u)$ in    \eqref{eq:main-conv-22} below, we further assume that $k_nT\to0$.
\end{Assumption}

We briefly comment on the various rate conditions in the above assumption. The first line of (\ref{eq:rel}) contains conditions that guarantee that the various biases in the recovery of the volatility increment are of higher asymptotic order relative to the rate of convergence in the CLT for the volatility of volatility and leverage effect estimators. The first of them is rather weak when $N$ is high. The second one is due to the biases of exact order $\sqrt{\Den}T$ that arise in an expansion of $\Delta^n_i V_{t,T}(u)$ (see the first line in (\ref{eq:incr-V-2})). If we use the bias-corrected estimators based on two maturities, the required condition is weaker than the one needed for estimators based on a single maturity, i.e., $k_nT=O(1)$ versus $k_nT=o(1)$. The third condition requires the length of the estimation window  to be not larger asymptotically than the time-to-maturity of the options, which will be the case in applications. 

The fourth condition makes sure that $VV_t-\frac{1}{k_n\Den}\int_{t-k_n\Den}^t VV_s ds$ and $LV_t-\frac{1}{k_n\Den}\int_{t-k_n\Den}^t LV_s ds$ (i.e., the biases arising from the difference between spot volatility of volatility / leverage effect and local averages thereof)  are asymptotically negligible. Finally, the last condition in the first line of (\ref{eq:rel}) is due to the Riemann approximation error of the integral in (\ref{eq:spanning}). This error is typically small for applications such as the one considered in our empirical analysis. Next, the first condition in the second line of (\ref{eq:rel}) is a balance condition between the size of the error in the estimation due to the diffusive component of the process $\sigma^2$ and the one due to the option observation error.\footnote{The rate of convergence of $\widehat{V}_{t,T}(u)$ for a fixed $t$ is $T^{1/4}/\sqrt{\delta}$, see \cite{T19}. } We note that we allow $\phi$ to take both values of $0$ and $1$, that is, we allow either of these two sources of error to dominate the other one. Finally, the requirements for the threshold $\upsilon_n$ in (\ref{eq:rel}) are mild. Indeed, 
$\upsilon_n=\infty$ (i.e., no truncation) is permitted.

A feasible CLT for the volatility of volatility estimators is given in the following theorem.	
\begin{theorem}\label{thm:vov}	Let $F:(0,\infty)\to\R$ be  a $C^2$-function and $\tau_n(x)=x\bone_{\{\lvert x\rvert\leq \upsilon_n\}}$.
		Suppose that the log-price process $x$, the volatility process $\si$, the observed option prices and the sequences $k_n$, $\Den$, $T$, $\delta$ and $\upsilon_n$ satisfy Assumptions~\ref{ass:C} and \ref{ass:hf} as well as Assumptions \ref{ass:main} and \ref{ass:main-1} from Appendix~\ref{sec:exp} (with the same $N\geq3$ as in Assumptions~\ref{ass:hf} and \ref{ass:main}).
	Further define
		\begin{equation}\label{eq:avar0} 
				\wh{\avar}^n_{t,T}(u)=	\wh{\avar}^{n,0}_{t,T}(u)+2\wh{\avar}^{n,1}_{t,T}(u),\quad \wh{\avar}^n_{t,T,T'}(u)=\wh{\avar}^{n,0}_{t,T,T'}(u)+2\wh{\avar}^{n,1}_{t,T,T'}(u),
		\end{equation}
	where
		\begin{equation}\label{eq:avar}\begin{split}
				\wh{\avar}^{n,0}_{t,T}(u)&=\frac{1}{k_n\Den^2}\Biggl(\sum_{i=2}^{k_n}\bigl(q^n_i(\wh V_{t,T}(u))\bigr)^2-\sum_{i=4}^{k_n}q^n_{i-2}(\wh V_{t,T}(u))q^n_i(\wh V_{t,T}(u))\Biggr),\\
				\wh{\avar}^{n,1}_{t,T}(u)&=\frac{1}{k_n\Den^2}\Biggl(\sum_{i=3}^{k_n} q^n_{i-1}(\wh V_{t,T}(u))q^n_i(\wh V_{t,T}(u))- \sum_{i=4}^{k_n}q^n_{i-2}(\wh V_{t,T}(u))q^n_i(\wh V_{t,T}(u))\Biggr), \\
				\wh{\avar}^{n,0}_{t,T,T'}(u)	&=\frac{1}{k_n\Den^2}\Biggl(\sum_{i=2}^{k_n}\bigl(q^n_i(\wh V_{t,T,T'}(u))\bigr)^2-\sum_{i=4}^{k_n}q^n_{i-2}(\wh V_{t,T,T'}(u))q^n_i(\wh V_{t,T,T'}(u))\Biggr),\\
				\wh{\avar}^{n,1}_{t,T,T'}(u)&=\frac{1}{k_n\Den^2}\Biggl(\sum_{i=3}^{k_n} q^n_{i-1}(\wh V_{t,T,T'}(u))q^n_i(\wh V_{t,T,T'}(u))- \sum_{i=4}^{k_n}q^n_{i-2}(\wh V_{t,T,T'}(u))q^n_i(\wh V_{t,T,T'}(u))\Biggr)
		\end{split}\raisetag{5\baselineskip}\end{equation}
		and
		\begin{align*}
			q^n_i(\wh V_{t,T}(u)) &= (\tau_n(\Delta_i^n\widehat{V}_{t,T}(u)))^2 + 2\tau_n(\Delta_{i-1}^n\widehat{V}_{t,T}(u))\tau_n(\Delta_i^n\widehat{V}_{t,T}(u)), 	 \\
			q^n_i(\wh V_{t,T,T'}(u)) &= (\tau_n(\Delta_i^n\widehat{V}_{t,T,T'}(u)))^2 + 2\tau_n(\Delta_{i-1}^n\widehat{V}_{t,T,T'}(u))\tau_n(\Delta_i^n\widehat{V}_{t,T,T'}(u)).
		\end{align*} 
		Then, for any $0<\underline t<t<\overline t<\infty$ and $u>0$, the estimators in \eqref{eq:vov} satisfy
			\begin{equation}\label{eq:main-conv-alt}  \begin{split}
				\sqrt{\frac{k_n}{\wh{\avar}^n_{t,T}(u)}}\bigl(\wh{VV}^n_{t,T}(u)-VV_t\bigr)&\xrightarrow{\mathcal{L}-s} N(0,1),\\	\sqrt{\frac{k_n}{\wh{\avar}^n_{t,T,T'}(u)}}\bigl(\wh{VV}^n_{t,T,T'}(u)-VV_t\bigr)&\xrightarrow{\mathcal{L}-s} N(0,1),\end{split}
		\end{equation}
		on  the set $\{\inf_{s\in[\underline t,\overline t]} \si_s^2 >0\}$,
		where $VV_t$ is given by \eqref{eq:VV} and the $N(0,1)$ limit variable is defined on a product extension of $(\ov \Om,\ov\calf,\ov\P)$  and is independent from it. 
	\end{theorem}
The corresponding CLT result for the leverage effect estimators is given in the next theorem.
\begin{theorem}\label{thm:lev} In the same set-up as in Theorem~\ref{thm:vov}, if we define
	\begin{equation}\label{eq:avar2}\begin{split}
			\wh{\avar}^{\prime n}_{t,T}(u)	&=\frac{1}{k_n\Den^2} \sum_{i=1}^{k_n} (\tau_n(\Delta^n_i \wh V_{t,T}(u)))^2(\tau_n(\Delta^n_i x_t))^2 \\
			&\quad-\frac{1}{k_n\Den^2} \sum_{i=3}^{k_n}\tau_n(\Delta^n_i \wh V_{t,T}(u))\tau_n(\Delta^n_{i}x_t)\tau_n(\Delta^n_{i-2} \wh V_{t,T}(u))\tau_n(\Delta^n_{i-2}x_t),\\
			\wh{\avar}^{\prime n}_{t,T,T'}(u) &=\frac{1}{k_n\Den^2} \sum_{i=1}^{k_n} (\tau_n(\Delta^n_i \wh V_{t,T,T'}(u)))^2(\tau_n(\Delta^n_i x_t))^2 \\
			&\quad-\frac{1}{k_n\Den^2} \sum_{i=3}^{k_n}\tau_n(\Delta^n_i \wh V_{t,T,T'}(u))\tau_n(\Delta^n_{i}x_t)\tau_n(\Delta^n_{i-2} \wh V_{t,T,T'}(u))\tau_n(\Delta^n_{i-2}x_t),
	\end{split}\end{equation}
	then
	\begin{equation}\label{eq:main-conv-22} \begin{split}
		\sqrt{\frac{k_n}{\wh{\avar}^{\prime n}_{t,T}(u)}}\bigl(\wh{LV}^n_{t,T}(u)-LV_t\bigr)&\xrightarrow{\mathcal{L}-s} N(0,1),\\	\sqrt{\frac{k_n}{\wh{\avar}^{\prime n}_{t,T,T'}(u)}}\bigl(\wh{LV}^n_{t,T,T'}(u)-LV_t\bigr)&\xrightarrow{\mathcal{L}-s} N(0,1),\end{split}\end{equation}
	where $LV_t$ was defined in \eqref{eq:VV} and the $N(0,1)$ limit variable is defined on a product extension of $(\ov \Om,\ov\calf,\ov\P)$  and is independent from it. 	
\end{theorem}

\begin{remark}\label{rem:avar}
From the proof of Theorems~\ref{thm:vov} and \ref{thm:lev} in the Appendix, one can show that $\sqrt{k_n}\frac{\Den}{\delta/\sqrt{T} + \Den}(\wh {VV}^n_{t,T}(u) - VV_t)$ and  $\sqrt{k_n}\frac{\Den}{\delta/\sqrt{T} + \Den}(\wh {VV}^n_{t,T,T'}(u)-VV_t)$ have asymptotic $\calf$-conditional variances given respectively by 
	\begin{align*}
	\avar(VV)^{(1)}_t(u)	&= 6VV_t^2\phi^2 + 8VV_t v_{t}^{(1)}(u) \phi(1-\phi)+40 (v_t^{(1)}(u))^2 (1-\phi)^2, \\
	\avar(VV)^{(2)}_t(u)	&= 6VV_t^2 \phi^2 + 8VV_t v_{t}^{(2)}(u)  \phi(1-\phi)+40 (v_t^{(2)}(u))^2 (1-\phi)^2, 
	\end{align*}
while $\sqrt{k_n}\sqrt{\frac{\Den}{\delta/\sqrt{T} + \Den}}(\wh {LV}^n_{t,T}(u) - LV_t)$ and $\sqrt{k_n}\sqrt{\frac{\Den}{\delta/\sqrt{T} + \Den}}(\wh {LV}^n_{t,T,T'}(u)-LV_t)$ have asymptotic $\calf$-conditional variances given respectively by 
\begin{align*}
	\avar(LV)^{(1)}_t(u)	&= (VV_t\si_t^2+LV_t^2)\phi  + 2\si_t^2v_{t}^{(1)}(u)  (1-\phi), \\
\avar(LV)^{(2)}_t(u)	&= (VV_t\si_t^2+LV_t^2)\phi + 2\si_t^2v_{t}^{(2)}(u)  (1-\phi).
\end{align*}
Here, $v_{t}^{(1)}(u) = v_{t,1}(u) $, $v_{t}^{(2)}(u) = (\frac{\tau}{\tau-1} )^2v_{t,1}(u)+ (\frac{1}{\tau-1} )^2v_{t,\tau}(u)$ ($\tau$ was defined in \eqref{eq:Tprime}),
\[ v_{t,\tau}(u) = 4e^{u^2\si_t^2} (F'(\si_t^2))^2\lvert \si_t\rvert^3\rho_{t,\tau}(0)\zeta_{t,\tau}(0)^2 \int_\R \cos^2(u\lvert\si_t\rvert k)\wt \Psi(k)^2 dk, \]
$\rho_{t,\tau}(k)$ and $\zeta_{t,\tau}(k)$ are the processes from Assumption~\ref{ass:C} and $\wt \Psi$ is defined in (\ref{eq:Psi_tilde}). 
\end{remark}

As evident from Remark~\ref{rem:avar}, the rate of convergence for all estimators is $\sqrt{k_n}$ if $\phi>0$, where $\phi$ is the constant in Assumption~\ref{ass:hf}. If $\phi = 0$, then the rate of convergence of the volatility of volatility estimators slows down to  $\sqrt{k_n}\Delta_n/(\delta/\sqrt{T})$ while that of the leverage effect estimators slows down to $\sqrt{k_n}\sqrt{\Delta_n/(\delta/\sqrt{T})}$.\footnote{Even though the volatility of volatility and the leverage effect estimators are not consistent if $\phi=0$ and $\sqrt{k_n}[\Den/(\Den+\delta/\sqrt{T})] \not\to \infty$, the limit results in Theorems~\ref{thm:vov} and \ref{thm:lev} continue to hold.} The case $\phi>0$ corresponds to the situation when option observation errors are sufficiently small so that their presence does not affect the rate of convergence of the estimators.\footnote{If $\phi\in(0,1)$, the option observation errors have an effect on the limit variance as evident from Remark~\ref{rem:avar}.} If this is not the case, i.e., if $\phi=0$,  the presence of option observation errors slows down  the rate of convergence of the estimators and determines their limit distributions.  The feasible CLTs in the above two theorems have the convenient feature, from an applied point of view,  that the user does not need to know a priori the value of $\phi$ from \eqref{eq:rel} above, i.e., which of the sources of estimation error is asymptotically dominant. Our estimates of the asymptotic variance are constructed in a way that adapts to the situation at hand.

We can compare the rate of convergence of $\wh{VV}^n_{t,T}(u)$ and $\wh{LV}^n_{t,T}(u)$ with that of their return-based counterparts. It is easiest to do so in the case when both the underlying asset price and the options written on it are not contaminated by observation errors (microstructure noise). In this case, the rate of convergence of the spot counterparts of the estimators of \cite{WM14}, \cite{vetter2015estimation}, \cite{AFLWY17} and \cite{KX17} for the volatility of volatility and the leverage effect using $k_n$  high-frequency    price observations is $k_n^{1/4}$. By contrast, our estimators $\wh{VV}^n_{t,T}(u)$ and $\wh{LV}^n_{t,T}(u)$ have a faster rate of convergence of $\sqrt{k_n}$ when there are no option observation errors (or the latter are not too big, i.e., when $\phi >0$). We will see in the Monte Carlo in Section~\ref{sec:mc} that the faster rate of convergence of the option-based estimators translates into rather nontrivial efficiency gains in finite samples over their return-based counterparts. This holds true even when the observed option prices contain measurement error, while the observed asset prices do not.      
	
\begin{remark}\label{rem:infvar} 
	While  price and volatility jumps are summable under Assumption~\ref{ass:main-1}, it is possible to extend Theorems~\ref{thm:vov} and \ref{thm:lev} to a setting where  both $x$ and $\si$ may have infinite variation jumps or arbitrary degree if one
	\begin{itemize}
		\item  assumes that the infinite variation jumps of $x$ and $\si$ are stable-like,
		\item replaces the truncated realized variance estimators in \eqref{eq:vov} by characteristic-function-based  estimators from \cite{jacod2014efficient} (see also \cite{LLL18}), and
		\item removes bias terms induced by infinite variation jumps by following the debiasing procedure in \cite{jacod2014efficient} (see also \cite{LLL18}).
	\end{itemize}
Such extensions, while conceptually easy, are rather tedious to present. We therefore leave the details of these modifications (including an examination of their performance in simulated and real data) to future research.
\end{remark}

We finish this section with summarizing in Table~\ref{table:estimators} existing estimators related to the estimation of volatility of volatility and the leverage effect. We only list nonparametric estimators with a feasible CLT in the table. We do not include the leverage effect estimators of \cite{andersen2015exploring} and \cite{KX17} based on the VIX index because they do not estimate $LV_t$ in general, which is what we are after. This is further discussed in Section~\ref{sec:mc} below. The first four entries in the table are estimators of spot volatility based on return or option data. 
The volatility of volatility and  leverage effect estimators from high-frequency asset price data given in the table are obtained by   integrating over time the spot quantities examined in this paper. It is relatively straightforward to extend these results to the case of spot volatility of volatility and spot leverage effect estimation. In addition, some of the existing leverage effect estimators are estimators of the correlation--and not covariance--between the diffusive price and the volatility. Again, these results  can be extended to the estimation of the leverage effect in the way  we define it here. 
 
\begin{table}[ht!]
	\setlength{\tabcolsep}{0.1cm}
	\begin{center}\small 
		\caption{{\bf Volatility-related Estimators}\label{table:estimators}}
		\begin{tabularx}{1.03\textwidth}{llcccl}
			\toprule
			{\bf Estimand}    & \multicolumn{1}{c}{\bf Data} & \multicolumn{2}{c}{\bf Jumps in} & \multicolumn{1}{c}{\bf Noise} & \multicolumn{1}{c}{\bf Outlet}\\[+1ex] 
			                           &                 & Price  & Volatility  & \\
			\multirow{2}{*}{Spot Volatility} & \multirow{2}{*}{high-frequency asset returns} & \multirow{2}{*}{yes} & \multirow{2}{*}{yes} & \multirow{2}{*}{yes} & \cite{ait2014high}\\
			                       &                                               &      &        &       & and references therein\\     
			Spot Volatility & high-frequency asset returns & yes & yes & no & \cite{LLL18}\\
			Spot Volatility & options & yes & yes & yes & \cite{T19}\\
			Spot Volatility & options & yes & yes & yes & \cite{todorov2021bias}\\[+1ex]	
			Volatility of Volatility & high-frequency asset returns & no & no & no & \cite{vetter2015estimation}\\
			Volatility of Volatility & high-frequency asset returns & yes & no & yes & \cite{li2022volatility}\\
			Volatility of Volatility & high-frequency option data & yes & yes & yes & current paper\\[+1ex]	
			Leverage Effect & high-frequency asset returns & yes & no & yes & \cite{WM14}\\	
			Leverage Effect & high-frequency asset returns &  yes & yes & yes & \cite{AFLWY17} \\
			Leverage Effect & high-frequency asset returns & yes & no & no & \cite{KX17}\\
			Leverage Effect & high-frequency asset returns & yes & yes & no & \cite{yang2023estimation}\\	
			\multirow{2}{*}{Leverage Effect} & high-frequency option data \& &  \multirow{2}{*}{yes}  & \multirow{2}{*}{yes} & \multirow{2}{*}{yes} & \multirow{2}{*}{current paper}\\
			                          & high-frequency asset returns & & & & \\

				\bottomrule
		\end{tabularx}
	\end{center}
\emph{Note}: The volatility of volatility and leverage effect estimators based on asset price returns in the cited papers  estimate integrated quantities.
\end{table}

\section{Monte Carlo Study}\label{sec:mc}

In this section, we evaluate the performance of the proposed volatility of volatility and leverage effect estimators on simulated data.
\subsection{Setup}
We use the following model for the underlying asset price $X_t=e^{x_t}$, under the risk-neutral probability measure $\Q$, to generate the true option prices:
\begin{equation}
	\frac{dX_t}{X_{t-}} = \sqrt{V_t}dW_t+\int_{\mathbb{R}}\left(e^x-1\right)\mu(dt,dx),
\end{equation}
where
\begin{equation}
	dV_t = \kappa_v(\theta_v-V_t)dt+\sigma_v\sqrt{V_t}dB_t,
\end{equation}
and $W_t$ and $B_t$ are $\mathbb{Q}$-Brownian motions with $\textrm{corr}(dW_t,dB_t) = \rho dt$, and $\mu$ is an integer-valued random measure with $\mathbb{Q}$-compensator $dt\otimes \nu_t(dx)$ and
\begin{equation}
	\nu_t(dx) = V_t (c_{-}e^{-\lambda_-|x|}1_{\{x<0\}} + c_{+}e^{-\lambda_+|x|}1_{\{x>0\}})dx.
\end{equation}
In the above specification for $X$, the stochastic variance is modeled as a square-root diffusion process like in the popular Heston model \citep{heston}. The price jumps have intensity that is affine in the level of diffusive variance like in \cite{DPS00} and subsequent empirical option pricing work. Our jump specification is a time-changed double-exponential model, with the time-change being the integrated diffusive variance. 

We consider three parameter settings for the above model. The parameter values for the three cases are given in Table~\ref{table:pars}.  In all of them, the unconditional mean of the variance is $\theta_v = 0.02$.  In the first specification, the volatility is very persistent with half-life of a shock to stochastic variance equal to six months. In the second and third specifications, the half-life of a shock to variance is one month and ten business days, respectively. In all cases, the parameter $\rho$ is set to $-0.9$ implying strong negative correlation between price and variance diffusive shocks. We note that the Feller condition ($\sigma_v^2\leq 2\kappa_v\theta_v$) puts an upper bound on $\sigma_v$ and this means that for more persistent dynamics the volatility of volatility is smaller. Turning next to the jump specification, we set  $\lambda_-=50$ and $\lambda_+=100$. This choice implies tail decays of out-of-the-money puts and calls similar to those of observed options written on the S\&P 500 index, see e.g., \cite{AFT,AFT_b}. Finally, we set $c_{\pm}$ according to 
\[c_- = 0.9\times \frac{\lambda_-^{3}}{2}~\textrm{and}~c_+ = 0.1\times \frac{\lambda_+^{3}}{2},\]
which implies that spot jump variation is equal to spot diffusive variance, and further that $90\%$ of the jump variation is due to negative jumps. This separation of the risk-neutral variation into a diffusive part and one due to positive and negative jumps is similar to that implied from parametric models fitted to observed S\&P 500 index options, see e.g., \cite{AFT_b}. 

For simplicity, the dynamics of $x$ under $\mathbb{P}$ is the same as that under $\mathbb{Q}$ with one exception. Mainly, we do not allow for jumps under $\mathbb{P}$. A more realistic specification would be one in which we allow for price jumps but with much smaller size than the one they have under $\mathbb{Q}$. We do not consider such an extension of the setup as it has only negligible effect on the results. 

\begin{table}[ht!]
	\setlength{\tabcolsep}{0.30cm}
	\begin{center}\small 
		\caption{{\bf Parameter Setting for the Monte Carlo}\label{table:pars}}
		\begin{tabularx}{\textwidth}{Xccccccccc}
			\toprule
			Case    & \multicolumn{4}{c}{Variance Parameters} & & \multicolumn{4}{c}{Jump Parameters} \\ 
			\cline{2-5} \cline{7-10}
			& $\theta_v$ & $\kappa_v$ & $\sigma_v$ & $\rho$ & & $\lambda_-$ & $\lambda_+$ & $c_-$ & $c_+$\\
			\hline
			S & $0.02$ & $1.39$ & $0.15$ & $-0.9$ & &  $50$ & $100$ & $3.6\times 10^3$ & $50\times 10^3$\\ 
			M & $0.02$ & $7.90$ & $0.40$ & $-0.9$ & &  $50$ & $100$ & $3.6\times 10^3$ & $50\times 10^3$\\ 
			F & $0.02$ & $17.50$ & $0.70$ & $-0.9$ & & $50$ & $100$ & $3.6\times 10^3$ & $50\times 10^3$\\
			\bottomrule
		\end{tabularx}
	\end{center}
\end{table}

Observed options are given by
\begin{equation}
	\widehat{O}_{t,T}(k_{t,T}(j)) = O_{t,T}(k_{t,T}(j))(1 + 0.015\times z_{t,T}(j)),\quad j=1,\dots,N_{t,T},
\end{equation}
where $\{z_{t,T}(j)\}_{j=1}^{N_{t,T}}$ are sequences of i.i.d.\ standard normal variables which are independent of each other. The size of the observation error is calibrated to roughly match  bid--ask spreads of index option data. We set $X_0 = 2500$ and $\Delta_n = 1/(252\times 80)$, which corresponds approximately to sampling option data every five minutes in a 6.5 hours trading day. We note that our unit of time is one year and we adopt a business time convention in which one business day is of length $1/252$. We set $k_n = 80$ which means that we use all intraday option data in the analysis. At each point in time, the strikes are multiples of $5$. The strikes below and above the current price are extended in both directions by increments of $5$ until the true out-of-the-money option price falls below $0.075$. This specification of the strike grid mimics that of available  S\&P 500 index options. Next, the value of the variance at the beginning of the local time window, $V_0$, is set to the 25th, 50th or 75th quantile of its marginal distribution. Finally, the short and long tenor of the options at the beginning of the time window are set to $T= 3/252$ and $T' = 6/252$, which correspond to $3$ and $6$ business days to expiration, respectively.

To reduce the Riemann sum approximation error of the integral in (\ref{eq:spanning}), we compute option prices on a strike grid with mesh $2.5$ from the observed ones using linear interpolation in Black--Scholes Implied Volatility (BSIV) space (recall that the strikes of observed options are multiples of $5$). 

For the computation of the volatility estimators, we need to choose the value of the characteristic exponent, i.e., the value of $u$ in $V_{t,T}(u)$. We do this in a data-driven way using the options at the first time point in the local window by setting $u$ to 
\begin{equation}\label{eq:uhat}
	\widehat{u}_{t^n_{k_n} ,T^n_{k_n}} =  \inf{\{u\geq 0: |\widehat{\mathcal{L}}_{t^n_{k_n},T^n_{k_n}}(u)|\leq 0.3\}}\wedge \textrm{argmin}_{u\in[0,\overline{u}_{t^n_{k_n},T^n_{k_n}}]}|\widehat{\mathcal{L}}_{t^n_{k_n},T^n_{k_n}}(u)|,
\end{equation}
where $\overline{u}_{t^n_{k_n},T^n_{k_n}} = \sqrt{-2\log(0.05)}/\widehat{\sigma}_{t^n_{k_n},ATM}$ and $\widehat{\sigma}_{t^n_{k_n},ATM}$ is the at-the-money Black--Scholes implied volatility at time $t^n_{k_n}$ for the shortest available maturity on that day. We thus look at $\widehat{V}_{t^n_i,T^n_i}(\widehat{u}_{t^n_{k_n},T^n_{k_n}})$ and $\widehat{V}_{t^n_i,T^n_{i},T^{\prime n}_{i}}(\widehat{u}_{t^n_{k_n},T^{\prime n}_{k_n}})$.
Finally, since the model does not feature volatility jumps, we do not perform truncation, i.e., we set $\upsilon_n = \infty$.  

\subsection{Return-Based Volatility of Volatility and Leverage Effect Estimators}
We compare the performance of our option-based estimators with their counterparts formed from high-frequency returns. We assume that we sample returns at a higher frequency than we sample option prices at. Using the higher sampling frequency, we then form estimates of volatility at the times $t_i^n$. More specifically, denote
\begin{equation}
\widehat{\mathcal{L}}_{t_i^n,ret}(u) = \frac{1}{l_n}\sum_{j=1}^{l_n}e^{iu(x_{t_i-j\Delta_n/l_n}-x_{t_i-(j-1)\Delta_n/l_n})/\sqrt{\Delta_n/l_n}},~~\widehat{\sigma}_{t_i^n,ret}(u) = -\frac{2}{u^2}\log|\widehat{\mathcal{L}}_{t_i^n,ret}(u)|,
\end{equation}
\begin{equation}
\widehat{V}_{t_i^n,ret}(u) = F(\widehat{\sigma}_{t_i^n,ret}(u)),~~\Delta_i^n \widehat{V}_{t,ret}(u) = \widehat{V}_{t_{i-1}^n,ret}(u) - \widehat{V}_{t_{i}^n,ret}(u),
\end{equation}
for some sequence $l_n\rightarrow\infty$. Using the above spot volatility estimates, we have the following return-based volatility of volatility and leverage effect estimators:
\begin{equation}\begin{split}
\wh{VV}^n_{t,ret}(u)&=\frac{1}{k_n\Den} \sum_{i=2}^{k_n-1}\Bigl((\tau_n(\Delta_i^n\widehat{V}_{t,ret}(u)))^2 + 2\tau_n(\Delta_{i-1}^n\widehat{V}_{t,ret}(u))\tau_n(\Delta_i^n\widehat{V}_{t,ret}(u))\Bigr), \\
\wh{LV}^n_{t,ret}(u)& = \frac{1}{k_n\Den} \sum_{i=1}^{k_n-2}\tau_n(x_{t_{i-1}^n} - x_{t_{i+1}^n})\tau_n(\Delta_i^n\widehat{V}_{t,ret}(u)).
\end{split}\end{equation}
We can further de-bias $\wh{VV}^n_{t,ret}(u)$ to account for the effect from the nonlinear transformation of $\widehat{\mathcal{L}}_{t_i^n,ret}(u)$ in forming the statistic (no such de-biasing is needed for $\wh{LV}^n_{t,T}(u)$). We do not do this for simplicity. We note also that we could have used alternatively local truncated volatility in constructing $\wh{VV}^n_{t,ret}(u)$ and $\wh{LV}^n_{t,ret}(u)$. We prefer the current estimators as they are the direct counterparts to the option-based ones. 

As for the option-based estimators, we set $u$ in a data-driven way. Towards this end, define
\begin{equation}
\begin{split}
RV_{t} &= \frac{1}{k_nl_n\Delta_n}\sum_{j=1}^{k_nl_n}(x_{t-(j-1)\Delta_n/l_n} - x_{t-j\Delta_n/l_n})^2,\\
BV_{t} &= \frac{\pi}{2}\frac{1}{k_nl_n\Delta_n}\sum_{j=1}^{k_nl_n-1}|x_{t-(j-1)\Delta_n/l_n} - x_{t-j\Delta_n/l_n}||x_{t-j\Delta_n/l_n} - x_{t-(j+1)\Delta_n/l_n}|,
\end{split}
\end{equation}
and set 
\begin{equation}
\widehat{u}_{t,ret} = \sqrt{-\frac{2\log(0.3)}{RV_t\wedge BV_t}}.
\end{equation}
With this choice of $u$, our return-based estimates become $\wh{VV}^n_{t,ret}(\widehat{u}_{t,ret})$ and $\wh{LV}^n_{t,ret}(\widehat{u}_{t,ret})$. We implement these estimators with $l_n = 72$. This corresponds to sampling every five seconds and forming local volatility estimates over blocks of length five minutes.

\subsection{Results}
The Monte Carlo results for estimating variance of log-variance (i.e., $VV_t$ in \eqref{eq:VV} with $F=\log x$) are reported in Table~\ref{tb:mc_vov}. We can draw several conclusions from them. First, the estimators which use options with one tenor only tend to be downward biased except for the very persistent case S. The size of this negative bias in most configurations increases with time to maturity of the options. This can be explained intuitively with higher-order biases present in the option-based spot volatility estimator, which increase in size as a function of the time to maturity. For example, the volatility of volatility generates a downward bias in the spot volatility estimator. This bias is naturally bigger for the third volatility specification (case F). As a result, the bias in our volatility of volatility estimators, using options with one tenor only, tends to be larger in magnitude for this volatility specification (reported in Panel C of the table). Second, the volatility estimator using two tenors corrects for the biases in spot volatility due to mean reversion in volatility and volatility of volatility. As a result, the estimator $\widehat{VV}_{0,T,T'}$ tends to have much smaller in magnitude biases in cases M and F. The reduction in bias is particularly large for the most volatile specification (case F). The cost of the reduction in bias in $\widehat{VV}_{0,T,T'}$ is increased volatility of that estimator when compared with the single-tenor estimators $\widehat{VV}_{0,T}$ and $\widehat{VV}_{0,T'}$. 

Third, volatility of volatility is most difficult to estimate when it is smallest, i.e., in scenario S. The reason for this is because in this case the signal to noise ratio is at its lowest. As a result, various biases in the estimation play a more important role in a relative sense now. In such a situation, the two tenor option-based estimator has comparable bias as its one-tenor counterparts. Finally, comparing the option-based estimates with the return-based ones, we can see the vast improvement offered by the option data for measuring  volatility of volatility. This is true regardless of the volatility specification and the starting value of volatility. Recall that our return-based estimator uses much higher sampling frequency than the option-based estimators. In spite of that, the option-based estimators are much more precise. This illustrates the advantages of using option data for inference related to volatility.\footnote{The performance of the return-based estimator will worsen further if one is to allow for market microstructure noise in the observed asset price.} 

\begin{table}[h!]                                                                                                   
	\centering                                                                                                                              
	\caption{Monte Carlo Results for Volatility of Volatility}\label{tb:mc_vov}                                                                                                                    
	\begin{tabularx}{1.0\textwidth}{X|rrr|rrr|rrr}                                                                                             
		\toprule                                                                                                                                
		Estimator &  Bias & STD & RMSE &  Bias & STD & RMSE & Bias & STD & RMSE \\         
		\hline     
		\multicolumn{10}{c}{\textbf{Panel A: Case S}}\\
	         \hline                                                                                                                             
		& \multicolumn{3}{c}{$V_0 = 0.0106 $} & \multicolumn{3}{c}{$V_0=0.0174$} & \multicolumn{3}{c}{$V_0=0.0265$} \\ 
		\hline                                                                                  
		$\widehat{VV}_{0,T}$ & 0.18 & 0.70 & 0.72 & -0.03 & 0.57 & 0.57 & 0.19 & 0.54 & 0.58 \\                             
		$\widehat{VV}_{0,T'}$ & 0.19 & 0.47 & 0.51 & -0.01 & 0.41 & 0.41 & 0.10 & 0.41 & 0.42 \\                                                    
		$\widehat{VV}_{0,T,T'}$ & 0.17 & 1.50 & 1.51 & 0.01 & 1.29 & 1.29 & 0.21 & 1.23 & 1.25 \\ 
		$\widehat{VV}_{0,ret}$ & 1.57 &158.28 & 158.20 & 14.24 & 224.41 & 224.75  & 19.23 & 322.00 & 322.37 \\                                             
		\hline     
		\multicolumn{10}{c}{\textbf{Panel B: Case M}}\\
		\hline                                                                                                                             
		& \multicolumn{3}{c}{$V_0 = 0.0095 $} & \multicolumn{3}{c}{$V_0=0.0167$} & \multicolumn{3}{c}{$V_0=0.0269$} \\ 
		\hline                                                                              
		$\widehat{VV}_{0,T}$ & -0.13 & 0.31 & 0.33 & -0.12 & 0.31 & 0.33 & -0.06 & 0.33 & 0.34 \\                             
		$\widehat{VV}_{0,T'}$ & -0.32 & 0.22 & 0.39 & -0.18 & 0.27 & 0.33 & -0.10 & 0.30 & 0.31 \\                                                    
		$\widehat{VV}_{0,T,T'}$ & 0.00 & 0.51 & 0.51 & -0.05 & 0.53 & 0.53 & -0.01 & 0.57 & 0.57 \\ 
		$\widehat{VV}_{0,ret}$ & 1.51 &28.62 & 28.66 & 0.07 & 33.49 & 33.49  & -0.65 & 41.61 & 41.62 \\   
		\hline  
		\multicolumn{10}{c}{\textbf{Panel C: Case F}}\\
		\hline   
		& \multicolumn{3}{c}{$V_0 = 0.0078$} & \multicolumn{3}{c}{$V_0=0.0156$} & \multicolumn{3}{c}{$V_0=0.0275$} \\   
		\hline        
		$\widehat{VV}_{0,T}$ & -0.39 & 0.23 & 0.46& -0.21 & 0.26 & 0.33 & -0.15 & 0.27 & 0.31 \\                             
		$\widehat{VV}_{0,T'}$ & -0.62 & 0.15 & 0.64 & -0.42 & 0.19 & 0.46 & -0.27 & 0.23 & 0.36 \\                                          
		$\widehat{VV}_{0,T,T'}$ & -0.11 & 0.43 & 0.44 & -0.01 & 0.40 & 0.40& -0.05 & 0.38 & 0.38 \\   
		$\widehat{VV}_{0,ret}$ &  0.89 & 10.37 & 10.40 & -0.17 & 11.43 & 11.43 & -0.51 & 13.23 & 13.24 \\                           
		\bottomrule                                                                                                                                                                                                                                                             
	\end{tabularx} 
	\smallskip
	\begin{minipage}{\textwidth}                                                                                                   
		\emph{Note}: Reported results are based on $1{,}000$ Monte Carlo replications. STD stands for standard deviation and RMSE for root mean squared error. All quantities are reported in relative terms, which is done by dividing by $VV_0$ (the true value). The top row of each panel reports the value of spot variance at the beginning of the time window.
	\end{minipage}                                                                                                                          
\end{table}            

We proceed next with a discussion of the Monte Carlo results for estimating the leverage effect corresponding to the log-variance (i.e., of $LV_t$ in \eqref{eq:VV} with $F=\log x$). These results are reported in Table~\ref{tb:mc_lev}. We note that for the Heston volatility model used in the Monte Carlo, $LV_t$ is a constant equal to $\rho\times\sigma_v$. Qualitatively, the results for estimating $LV_t$ are similar to those for estimating $VV_t$. In particular, the biases in the estimation are larger in magnitude for the fastest mean reversion simulation scenario F. The use of two tenors is rather beneficial in terms of reducing the downward bias in the estimation. Comparing the results in Tables~\ref{tb:mc_vov} and \ref{tb:mc_lev}, we can see that the precision of the estimation of the leverage effect is higher than that of the volatility of volatility. Finally, the option-based estimators again perform much better than the return-based one. The advantage in relative sense is not as big as in the case of volatility of volatility but is nevertheless quite nontrivial.

\begin{table}[h!]                                                                                                   
	\centering                                                                                                                              
	\caption{Monte Carlo Results for Leverage Effect}\label{tb:mc_lev}                                                                                                                    
	\begin{tabularx}{\textwidth}{X|rrr|rrr|rrr}                                                                                             
		\toprule                                                                                                                                
		Estimator &  Bias & STD & RMSE &  Bias & STD & RMSE & Bias & STD & RMSE \\   
	         \hline     
		\multicolumn{10}{c}{\textbf{Panel A: Case S}}\\
		\hline                                                                                                                             
		& \multicolumn{3}{c}{$V_0 = 0.0106 $} & \multicolumn{3}{c}{$V_0=0.0174$} & \multicolumn{3}{c}{$V_0=0.0265$} \\ 
		\hline  
		$\widehat{LV}_{0,T}$ & -0.05 & 0.29 & 0.29 & -0.02 & 0.27 & 0.27 & 0.00 & 0.24 & 0.24 \\                             
		$\widehat{LV}_{0,T'}$ & -0.03 & 0.22 & 0.22 & 0.00 & 0.22 & 0.22 & 0.00 & 0.21 & 0.21 \\                                                    
		$\widehat{LV}_{0,T,T'}$ & -0.07 & 0.42 & 0.43 & -0.02 & 0.38 & 0.38 & 0.00 & 0.36 & 0.36 \\ 
		$\widehat{LV}_{0,ret}$ & 0.07 & 4.12 & 4.11 & 0.22 & 4.57 & 4.58 & -0.01 & 5.43 & 5.43 \\                                                        
		\hline     
		\multicolumn{10}{c}{\textbf{Panel B: Case M}}\\
		\hline                                                                                                                             
		& \multicolumn{3}{c}{$V_0 = 0.0095 $} & \multicolumn{3}{c}{$V_0=0.0167$} & \multicolumn{3}{c}{$V_0=0.0269$} \\ 
		\hline  
		$\widehat{LV}_{0,T}$ & -0.07 & 0.18 & 0.20 & -0.07 & 0.19 & 0.20 & -0.04 & 0.18 & 0.19 \\                             
		$\widehat{LV}_{0,T'}$ & -0.18 & 0.16 & 0.24 & -0.10 & 0.17 & 0.19 & -0.06 & 0.17 & 0.18 \\                                                    
		$\widehat{LV}_{0,T,T'}$ & -0.01 & 0.24 & 0.24 & -0.04 & 0.25 & 0.25 & -0.02 & 0.24 & 0.24 \\ 
		$\widehat{LV}_{0,ret}$ &-0.04 & 1.75 & 1.75 & -0.08 & 1.84 & 1.84 & -0.08 & 2.00 & 2.00 \\             
		\hline  
		\multicolumn{10}{c}{\textbf{Panel C: Case F}}\\
		\hline   
		& \multicolumn{3}{c}{$V_0 = 0.0078$} & \multicolumn{3}{c}{$V_0=0.0156$} & \multicolumn{3}{c}{$V_0=0.0275$} \\   
		\hline        
		$\widehat{LV}_{0,T}$ & -0.22 & 0.16 & 0.27 & -0.11 & 0.17 & 0.20 & -0.08 & 0.17 & 0.18 \\                             
		$\widehat{LV}_{0,T'}$ & -0.38 & 0.13 & 0.40 & -0.23 & 0.14 & 0.27 & -0.14 & 0.15 & 0.21 \\                                          
		$\widehat{LV}_{0,T,T'}$ & -0.06 & 0.21 & 0.22 & -0.01 & 0.21 & 0.21 & -0.03 & 0.20 & 0.20 \\
		$\widehat{LV}_{0,ret}$ &  -0.08 & 1.04 & 1.04  & -0.08 & 1.08 & 1.08 & -0.07 & 1.14 & 1.14 \\	                            
		\bottomrule                                                                                                                                                                                                                                                             
	\end{tabularx} 
	\smallskip
	\begin{minipage}{\textwidth}                                                                                                   
		\emph{Note}: Reported results are based on $1{,}000$ Monte Carlo replications. STD stands for standard deviation and RMSE for root mean squared error. All quantities are reported in relative terms, which is done by dividing by $LV_0$ (the true value). The top row of each panel reports the value of spot variance at the beginning of the time window.
	\end{minipage}                                                                                                                          
\end{table}

\subsection{Using Alternative Volatility Proxies}                                                                                                                            

We finish this section with a comparison of the estimates of volatility of volatility and the leverage effect when one uses the VIX index as a proxy for the spot volatility as done by \cite{andersen2015exploring} and \cite{KX17}. The theoretical value of the squared VIX index is $\frac{1}{\mathcal{T}}\mathbb{E}_t\Bigl[\int_t^{t+\mathcal{T}}\sigma_s^2ds + 2\int_t^{t+\mathcal{T}}\int_{\mathbb{R}}(e^z-1-z)\mu(ds,dz)\Bigr]$. For the model used in the Monte Carlo, the value of the square of the VIX index is given by
\begin{equation}\label{eq:vix}
VIX_t^2 = \left(1+0.9\frac{\lambda_-}{\lambda_- -1} +0.1\frac{\lambda_+}{\lambda_+ +1}\right)\left[V_t\left(\frac{1-e^{-\kappa_v\mathcal{T}}}{\kappa_v\mathcal{T}}\right) + \theta_v\left(1 - \frac{1-e^{-\kappa_v\mathcal{T}}}{\kappa_v\mathcal{T}}\right)\right],~~\mathcal{T} = \frac{1}{12}.
\end{equation}
This implies that the ratio of the diffusion coefficients of $\log(VIX_t^2)$ and $\log(\sigma_t^2)$ in the model is given by
\[ s_t = \frac{2V_t}{VIX_t^2}\frac{1-e^{-\kappa_v\mathcal{T}}}{\kappa_v\mathcal{T}}. \]
Therefore, the counterpart of $VV_t$ for $VIX_t$ is $s_t^2\times VV_t$ and of $LV_t$ is $s_t\times LV_t$. As an example, the ratio $s_t$ takes values in the range $(0.3,\,0.6)$ for scenario F, depending on the starting value of the volatility considered here. This implies large differences in volatility of volatility and the leverage effect depending on whether one uses $\sigma_t$ or $VIX_t$. Note also that this gap is time-varying implying that the dynamics of the series will be different as well.       

\section{Empirical Illustration}\label{sec:emp}
\subsection{Data}
We apply the developed estimation techniques to study the volatility of volatility and leverage effect of the S\&P 500 market index. We use high-frequency data of S\&P 500 index options traded on the CBOE options exchange for the period 2016--2020. On each trading day in our sample period, we record option prices every five minutes, starting from 9.35 AM EST until 4.00 PM EST. At each point in time we record out-of-the-money option mid-quotes, keeping only strikes with nonzero bids and with a ratio of ask/bid of less than $10$. The moneyness of the options is determined on the basis of a synthetic futures extracted from the option data using put-call parity and short-term T-bill interest rate. We remove from the analysis maturities for which the gap between strikes of the three nearest-the-money puts and calls is above $5$ (which is the minimum strike gap for SPX options). We remove also maturities for which $\max
\{\widehat{O}_{t,T}(k_{t,T}(1)),\widehat{O}_{t,T}(k_{t,T}(N_{t,T}))\}/\max_{j=1,...,N_{t,T}}\widehat{O}_{t,T}(k_{t,T}(j))>0.025$. 

From the remaining maturities on a given trading day, we keep the two closest to expiration with time to maturity between $2$ and $16$ business days. We require that the gap between the two tenors be at least three business days. The median time-to-maturities $T$ and $T'$ at the end of each trading day in our sample  are approximately $2$ and $6$ business days (and in $95\%$ of the days in the sample we use only options that have less than two weeks to expiration). These numbers are very close to the ones used in the Monte Carlo. In choosing $T$ and $T'$, there is a bias-variance trade-off. Smaller $T$ and $T'$ imply less biased spot volatility estimators (as our asymptotic expansions are for $T,T'\rightarrow 0$) but more noisy ones as we have fewer option observations per maturity at each point in time as $\delta$ (the mesh of the strike grid) is fixed. Our Monte Carlo experiments showed that, for $T$ and $T'$ similar to the ones used in the empirical application, the finite sample performance of the estimators is good.  

Finally, we exclude from the analysis February 5 and 6, 2018, during which there were wild intraday movements in volatility associated with the collapse of inverse exchange-traded volatility instruments. The choice of the characteristic exponent $u$ is done  exactly as in our simulation expirement. As in the Monte Carlo, we compute option prices on an equidistant strike grid with strike gap of $2.5$ via linear interpolation in BSIV space from the observed ones. We perform truncation of the volatility increments by setting the truncation parameter $\upsilon_n$ to
$$\widehat{\upsilon}_n = 3\left(\frac{\pi}{8}\frac{1}{(k_n-1)\Delta_n}\sum_{l=0}^{3}\sum_{i=2}^{k_n}|\Delta_i^n\widehat{V}_{t-l/252}|\,|\Delta_{i-1}^n\widehat{V}_{t-l/252}|\right)^{0.49},$$
where $\widehat{V}_t$ is one of the three considered spot variance estimators. 

\subsection{Results}
In the middle panel of Figure~\ref{fig:sp_vov_lev}, we plot our estimates of the market volatility of volatility, i.e., our estimates of the square of the diffusive coefficient of $\log(\sigma_t^2)$. In the top panel of the figure, we also plot the market spot diffusive volatility. We display $5$-day moving averages of the daily estimates in order to focus on lower frequency variation in the time series and minimize the impact of measurement error. As seen from the figure, market volatility of volatility displays only mild variation during our sample with occasional and relatively short-lived spikes. We plot both $\widehat{VV}_{t,T}$ as well as $\widehat{VV}_{t,T,T'}$. Exactly as in the Monte Carlo, $\widehat{VV}_{t,T,T'}$ is higher than $\widehat{VV}_{t,T}$. The gap between the two series is particularly large during 2017 when volatility was relatively low. Recall from the Monte Carlo that when volatility is low (relative to its mean), the downward bias in $\widehat{VV}_{t,T}$ tends to be more significant. 

\begin{figure}[h!]
	\begin{center}
		\includegraphics[scale=0.75]{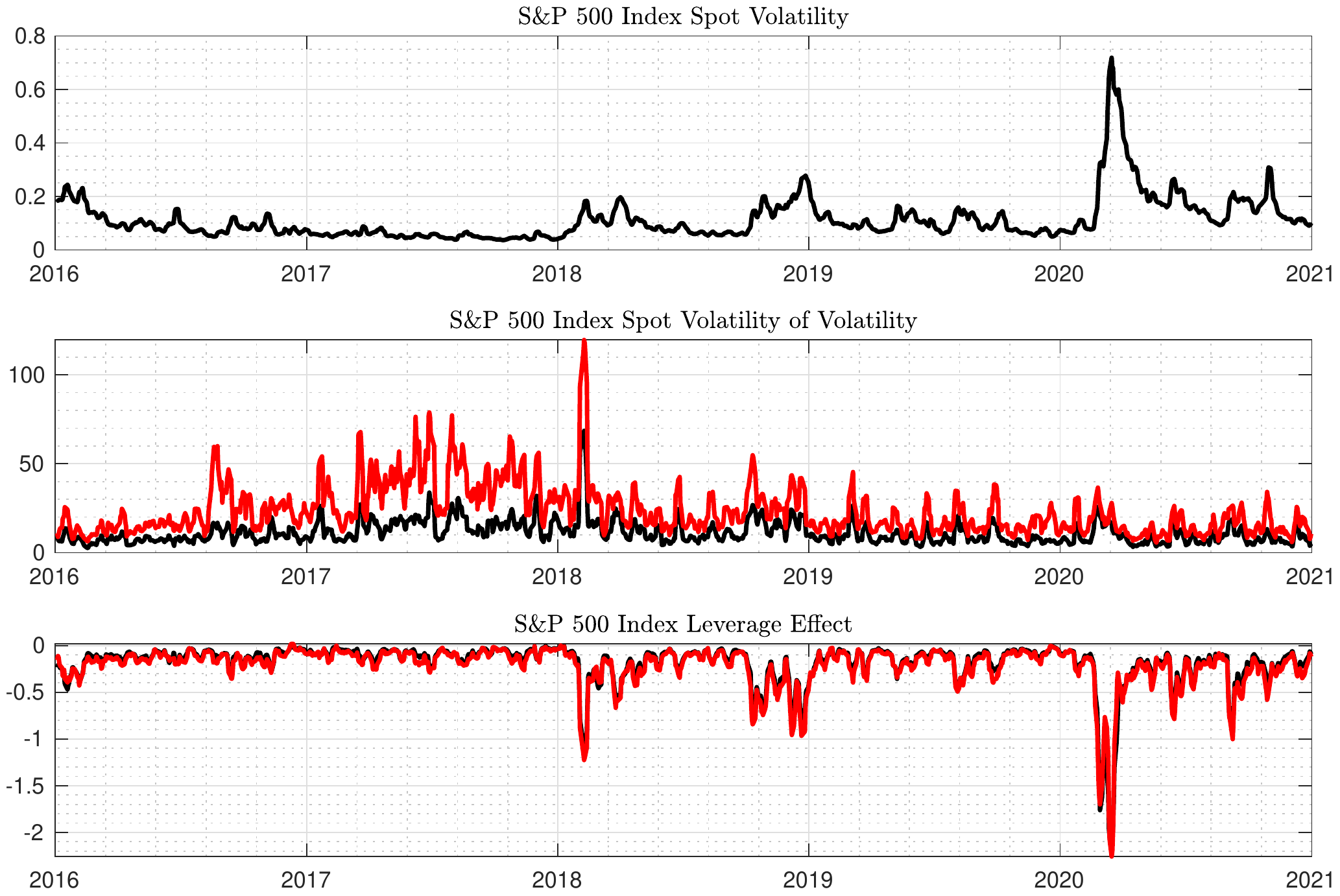}
	\end{center}
	\caption{{\bf S\&P 500 Index Volatility Risk}. The top panel displays a $5$-day moving average of annualized $\sqrt{\widehat{V}_{t,T}(u)}$; the middle panel displays $5$-day moving averages of $\widehat{VV}_{t,T}$ (black line) and $\widehat{VV}_{t,T,T'}$ (red line); the bottom panel displays $5$-day moving averages of $\widehat{LV}_{t,T}$ (black line) and $\widehat{LV}_{t,T,T'}$ (red line).}\label{fig:sp_vov_lev} 
\end{figure}

Related to that, the reported volatility of spot volatility estimates are somewhat larger than the values of the VVIX index during the same period (when reported in the same units as our estimators). The reason for this is that VVIX is a measure of volatility of the VIX index, which is a conditional expectation of one-month future volatility. The mean-reversion in volatility makes the VIX smoother than the spot volatility and this can explain the differences between our estimates and the VVIX, see equation (\ref{eq:vix}) above and the discussion afterwards. 

Consistent with earlier work cited in the introduction that argues   for a partial disconnect between volatility risk and volatility of volatility risk, we find weak correlation between $\widehat{V}_{t,T}(u)$ and $\widehat{VV}_{t,T,T'}$. For example, the big increase in market volatility in the Spring of 2020 is not accompanied by a significant change in the volatility of volatility (recall that $\widehat{VV}_{t,T,T'}$ is an estimate of volatility of log-variance).   

In the bottom panel of Figure~\ref{fig:sp_vov_lev}, we plot our estimates of a $5$-day moving average of the market leverage effect. Unlike the case of volatility of volatility, the two estimates $\widehat{LV}_{t,T}$ and $\widehat{LV}_{t,T,T'}$ are now much closer on average. We note also that the leverage effect increases sharply in magnitude at the onset of the pandemic in the Spring of 2020 when the market volatility also increased a lot.  

As mentioned in the introduction, our nonparametric estimates for the volatility of volatility and the leverage effect can be used as diagnostic tools for the specification of  volatility dynamics. To illustrate this, we can contrast the time series properties of our nonparametric estimates with those implied by  two popular volatility specifications considered in prior work. One is the Heston model that we used in our Monte Carlo, with the possible addition of volatility jumps, see e.g., \cite{DPS00}. The other is an exponential Ornstein--Uhlenbeck (OU) volatility specification, see e.g., \cite{chernov2003alternative}, in which the logarithm of the diffusive variance follows an Ornstein--Uhlenbeck process. 

The Heston model plus jumps and the exponential Ornstein--Uhlenbeck (OU) model differ in terms of the volatility of volatility and leverage effect they generate. For the Heston model, $VV_t$ is proportional to $1/\sigma_t^2$ while $LV_t$ is constant. On the other hand, for the exponential-OU volatility model, $VV_t$ is constant while $LV_t$ is proportional to $\sigma_t$. Our nonparametric evidence shows that the exponential-OU volatility model can better rationalize the observed dynamics of $VV_t$ and $LV_t$. In particular, the estimated volatility of volatility exhibits little time series variation while the estimated leverage effect is inversely related to the market volatility. These features of the observed estimates are at odds with those implied by the classical Heston model. 

\section{Conclusion}\label{sec:concl}

The volatility of volatility and leverage effect are asset price characteristics that are notoriously difficult to estimate from high-frequency price data, particularly when price and volatility have jumps and market microstructure noise is present. This is mainly due to the latency of volatility and the fact that return-based spot volatility estimates are rather noisy. 

In this paper, we propose to estimate the volatility of volatility and the leverage effect using high-frequency data from short-dated options written on the underlying asset. Our nonparametric estimators are consistent, asymptotically mixed normal and permit feasible inference in the joint presence of price and volatility jumps and  observation errors. A simulation study reveals nontrivial efficiency gains from using the proposed estimators over ones constructed using only high-frequency price data. 

Applying  our estimators to S\&P 500 index high-frequency option data, we generate  reliable market spot volatility of volatility and leverage effect estimates using as little as one day of $5$-minute data. Our estimates provide nonparametric evidence in favor of an exponential-OU model and against a Heston-type model for market volatility. 


\appendix

\section*{Appendix}
\section{High-Frequency Expansions of Volatility Estimators Based on Characteristic Functions}\label{sec:exp}

The estimators constructed in the main text rely on asymptotic expansions of $\Delta^n_i V_{t,T}$ and $\Delta^n_i V_{t,T,T'}$ from \eqref{eq:incr-V}. Such expansions are obtained in \cite{CT23_a} under the assumption that the price process $x$ is a  \emph{deep It\^o semimartingale with $N$ hidden layers}, where $N$ in our case is the same number as in \eqref{eq:rel}.  Informally speaking, $x$ is a deep It\^o semimartingale with $N$ hidden layers, if its coefficients (i.e., the drift $\al$, the volatility $\si$, and its jump intensity, to be introduced formally below) are again It\^o semimartingales  (layer 1), which in turn have coefficients that are  It\^o semimartingales (layer 2) etc., repeated in this way until the $N$th layer. 

In order to state the formal definition, we
fix two integers $d,d'\geq2$ (which may depend on $N$). The integer $d$ is the number of Brownian motions needed to model the joint diffusive behavior of $x$ and of all coefficients in all $N$ layers. Accordingly, we let $\mathbb{W}=(W^{(1)},\dots,W^{(d)})^\top$  be a $d$-dimensional standard $\F$-Brownian motion   such that $W^{(1)}=W$ and $W^{(2)}=\ov W$. The integer $d'$, which may or may not be equal to $d$, is used to model the joint jump behavior of $x$ and its ingredients. 
Thus, we consider an integer-valued random measure $\mu$, whose $\F$-compensator is given by 
\begin{equation}\label{eq:nu} 
	\nu(ds,dz) = \la(s-,z)F(dz) ds
\end{equation}
for some intensity process $\la(s,z)$ and $d'$-dimensional measure $F$. One way of realizing $\mu$ is to start with an $\F$-Poisson random measure $\pf$ on $[0,\infty)\times\R^{d'}\times\R$ (independent of $\mathbb{W}$) with intensity measure $\qf(dt,dz,dv)=dt F(dz) dv$ and then to define
\begin{equation}\label{eq:mu} 
	\mu(ds,dz) = \int_{\R} \bone_{\{ 0\leq v\leq \la(s-,z) \}} \pf(ds, dz,dv).
\end{equation}
With $\mu$ at hand, we assume that the jump parts of $x$ and $\si$ from \eqref{eq:x0} and \eqref{eq:si0} are respectively given by 
\begin{equation}\label{eq:jumps} 
	J^x_t = \iint_0^t  \ga(s,z)( \mu-\nu)(ds,dz),\qquad J^\si_t = \iint_0^t  \ga^\si(s,z)(\mu-\nu)(ds,dz),
\end{equation}
where $\ga$ and $\ga^\si$ are some predictable processes and  $\iint_s^t = \int_s^t \int_{\R^{d'}}$.  

The intensity process $\la(t,z)$ from \eqref{eq:nu} is assumed to be a nonnegative It\^o semimartingale (for fixed $z$) of the form
\begin{equation}\label{eq:la} \begin{split}
		\la(t,z)&=\la(0,z)+\int_0^t\al^\la(s,z)ds+\sum_{i=1}^d 	\int_0^t\si^{\la,(i)}(s,z) dW^{(i)}_s \\
		&\quad+ \iiint_0^t  \ga^\la(s,z,z',v') (\pf-\qf)(ds,dz',dv')
		,\end{split}
\end{equation}
where  $\iiint_s^t = \int_s^t\int_{\R^{d'}}\int_{\R}$ and $\al^\la$, $\si^{\la,(i)}$ and $\ga^\la$ are predictable coefficients.

Equations~\eqref{eq:x0} and \eqref{eq:si0}  specify the log-price process $x$ and its associated volatility process $\si$. Since $x$ is assumed to be a deep It\^o semimartingale, we still need to specify the other coefficients appearing in the various layers of $x$. To do so in a concise way, we use matrix notation and define
\begin{align*} 
	\theta(s,z)&= (\si_s,0,\dots,0,\al_s,\ga(s,z)) \in \R^{1\times (d+2)},	 \\
	y(ds,dz)&=(dW^{(1)}_s\delta_0(dz),\dots,dW^{(d)}_s \delta_0(dz), ds \delta_0(dz),( \mu-\nu)(ds,dz))^\top, 
\end{align*}
where $\delta_0$ is the Dirac measure at $0$. Note that $\theta$ is a row vector, while $y$ is a column vector (of measures). So using matrix-notation, we can  compactly write
\begin{equation}\label{eq:xshort} 
	x_t=x_0 +\iint_0^t \theta(s,z)y(ds,dz)
\end{equation} 
instead of \eqref{eq:x0}. In order to define the semimartingale characteristic of $\theta$ (and the semimartingale characteristics thereof, and so on), we use tensors (or arrays), adding one dimension each time we go one layer deeper. More precisely, 
for every $i=1,\dots,N$, we recursively  define processes $\{\theta(t,z_1,\dots,z_i):t\geq0, z_1,\dots,z_i\in\R^{d'}\}$ with values in $\R^{1\times (d+2)\times\cdots\times (d+2)}=\R^{1\times (d+2)^{\times i}}$ through
\begin{equation}\label{eq:beta} 
	\theta(t,z_1,\dots,z_i) = \theta(0,z_1,\dots,z_i)+\iint_0^t \theta(s,z_1,\dots,z_i,z_{i+1}) y(ds,dz_{i+1}),
\end{equation}
where for $A\in \R^{k_1\times\dots\times k_i}$ and $v\in\R^{k_i}$, the product $Av\in\R^{k_1\times\dots\times k_{i-1}}$ is given by
\[ (Av)_{j_1,\dots,j_{i-1}} = \sum_{j_i=1}^{k_i} A_{j_1,\dots,j_i}v_{j_i},\quad j_\ell=1,\dots,k_\ell,\quad \ell=1,\dots,i-1. \]
In other words, $\theta(t,z_1,z_2)\in\R^{1\times (d+2)\times (d+2)}$ specifies the semimartingale coefficients of $\theta(t,z)$, $\theta(t,z_1,z_2,z_3)\in\R^{1\times (d+2)\times (d+2)\times(d+2)}$ specifies the semimartingale coefficients of $\theta(t,z_1,z_2)$, etc. For example, if $x$ and $\si$ are given by \eqref{eq:x0} and \eqref{eq:si0}, respectively, then 
\begin{equation}\label{eq:x-1} 
		\theta(t,z)_1=\si_t, \quad \theta(t,z)_2=\cdots=\theta(t,z)_{d}=0, \quad \theta(t,z)_{d+1}=\al_t,\quad
		\theta(t,z)_{d+2}=\ga(t,z),
\end{equation}
and
\begin{equation}\label{eq:si-1} 
	\begin{aligned}
		\theta(t,z_1,z_2)_{21} &= \si^\si_t,\quad \theta(t,z_1,z_2)_{22}=\ov \si^\si_t,\quad
		\theta(t,z_1,z_2)_{23}=\cdots=\vartheta(t,z_1,z_2)_{2,d}=0,\\
		\theta(t,z_1,z_2)_{2,d+1}	&=\al^\si_t,\quad \theta(t,z_1,z_2)_{2,d+2}=\ga^\si(t,z_2)
	\end{aligned}
\end{equation}
Note that if $\bj_i=(j_1,\dots,j_i)$ and $j_\ell\in\{1,\dots,d+1\}$ for some $\ell\in\{1,\dots,i\}$, then only the value of $\theta(t,\bz_i)_{\bj_i}$ at $z_\ell=0$ matters, because  the spatial variable $z_\ell$ only matters for the jump part but not for drift or volatility. For instance, in the above example,  there is no loss of generality if we assume  $\theta(t,z_1,z_2)_{2,d+2}=\theta(t,0,z_2)$ ($=\ga^\si(t,z_2)$). 

We are now in the position to state the structural assumptions on the price process $x$ needed for our asymptotic analysis. We write $\bz_i= (  z_1,\dots,z_i)$ in what follows.
 \begin{Assumption}\label{ass:main} Under the risk-neutral probability measure $\Q$, the logarithmic price process $x$ is a  deep It\^o semimartingale with $N\geq3$ hidden layers given by \eqref{eq:xshort} such that there exist a localizing sequence $(T_n)_{n\in\N}$ of stopping times, an exponent $r\in[1,2)$, deterministic nonnegative measurable functions $J_n(z)$, $j_n(z,v)$ and $\mathcal{J}_n(\bz_{N+1})$, and for all $0<t<t'<t+1$, an $\F$-predictable process $(s,z)\mapsto \la_{t,t'}(s,z)$ defined for $s\in [t,t+1]$ and $z\in\R^{d'}$ with the following properties:
 	\begin{enumerate}
 		\item The functions $J_n(z)$ and $j_n(z,v)$ are real-valued and   $\int_{\R^{d'}\times\R} j_n(z,v) F(dz)dv<\infty$ for each $n\in\N$. Moreover, for all $t<T_n$, $t'\in[t,t+1]$, $z,z'\in\R^{d'}$ and $v'\in\R$, 
 		\begin{equation}\label{eq:bound} 
 			\lvert\la(t,z)\rvert+	 \lvert \al^\la(t,z)\rvert+\sum_{i=1}^{d} \lvert \si^{\la,(i)}(t,z)\rvert \leq J_n(z),\quad \lvert \ga^\la(t,z,z',v')\rvert^2
 			\leq J_n(z)^2j_n(z',v')
 		\end{equation}
 		and
 		\begin{equation}\label{eq:cont}\begin{split}
 				\sum_{i=1}^d   \E[(\si^{\la,(i)}(t',z)-\si^{\la,(i)}(t,z))^2\wedge1]^{1/2}& \leq \lvert t'-t\rvert^{1/2}J_n(z),	 \\
 				\E[(\ga^\la(t',z,z',v')-\ga^\la(t,z,z',v'))^2\wedge1]^{1/2}	&\leq \lvert t'-t\rvert^{1/2}J_n(z)	j_n(z',v')^{1/2}.
 		\end{split}\end{equation}
 		\item The function $\mathcal{J}_n(\bz_{N+1})$ takes values in $\R^{(d+2)^{N+1}}$ and for all $n\in\N$, $j_1,\dots,j_{N+1}\in\{1,\dots,d+2\}$, $z_1,\dots,z_{N+1} \in\R^{d'}$ and $t<T_n$, we have that
 		\begin{equation}\label{eq:theta-bound} 
 			\lvert\theta(t,\bz_{N+1})_{j_1,\dots,j_{N+1}}\rvert^r \vee 	\lvert\theta(t,\bz_{N+1})_{j_1,\dots,j_{N+1}}\rvert^2 \leq \mathcal{J}_n(\bz_{N+1})_{j_1,\dots,j_{N+1}},
 		\end{equation}
 		\begin{equation}\label{eq:beta-L2} 
 			\int_{(\R^{d'})^{N+1}} 
 			\mathcal{J}_n(\bz_{N+1})_{j_1,\dots,j_{N+1}}
 			\prod_{\ell: j_\ell \neq d+2} \delta_0(dz_{\ell})\prod_{\ell: j_\ell=d+2} J_n(z_\ell)  F(dz_{\ell}) < \infty. 
 		\end{equation}
 		\item For any $0<t<t',s<t+1$ and $z\in\R^{d'}$, the random variable $\la_{t,t'}(s,z)$ is $\calf_t\vee \calg_{t'}$-measurable. Moreover, if $s<T_n$, then
 		\begin{equation}\label{eq:la-smooth} 
 			\E[(\la(s,z)-\la_{t,t'}(s,z))^2\wedge1]^{1/2}\leq (\lvert t'-t\rvert^{1/2}+ [( s-t')_+]^2) J_n(z),
 		\end{equation}
 		where $x_+ = x\vee 0$.
 	\end{enumerate} 
 	Still under $\Q$, we also assume  \eqref{eq:x0} and \eqref{eq:si0} without loss of generality, so that we have \eqref{eq:x-1} and \eqref{eq:si-1}.
 \end{Assumption}
 
 Assumption \ref{ass:main} guarantees that the integrals in \eqref{eq:xshort} are all well defined and It\^o semimartingales, and in fact, is not much stronger than just that. Part 1 and 2 contain some classical  integrability and regularity conditions on the last layer of the deep It\^o semimartingale $x$ as well as the coefficients of the intensity process $\la$. Part 3 of Assumption~\ref{ass:main} is also mild and is satisfied, for example, if the intensity process $\la(t,z)$ is also a deep It\^o semimartingale with at least three hidden layers. Consider, for instance, the cases where  $\la(s,z)=\la_s z$ and
 \begin{align*}
 	\la_s&=\la_0+\int_0^s \si^\la_r dW_r,	\quad\si^\la_r=\si^\la_0+\int_0^r \si(\si^\la)_u dW_u,\\
 	\si(\si^\la)_u&=\si(\si^\la)_0+\int_0^u \si(\si(\si^\la))_v dW_v,\quad \si(\si(\si^\la))_v=\si(\si(\si^\la))_0+\int_0^v \si(\si(\si(\si^\la)))_w dW_w
 \end{align*}
 for some locally bounded $\si(\si(\si(\si^\la)))$.
 Then one can choose
 \begin{align*}
 	\la_{t,t'}(s,z)&=\biggl(\la_t+\si_t^\la(W_s-W_{t'})+\si(\si^\la)_t\int_{t'}^s\int_{t'}^r dW_udW_r+\si(\si(\si^\la))_t\int_{t'}^s\int_{t'}^r\int_{t'}^u dW_v dW_udW_r\\
 	&\quad+\si(\si(\si(\si^\la)))_t\int_{t'}^s\int_{t'}^r\int_{t'}^u \int_{t'}^v dW_wdW_v dW_udW_r\biggr)z,
 \end{align*} 
 which yields $\la(s,z)-\la_{t',t'}(s,z)=O((s-t')^2)$ and $\la_{t',t'}(s,z)-\la_{t,t'}(s,z)=O(\sqrt{t'-t})$ and hence \eqref{eq:la-smooth}.
 Clearly, this example easily extends to   cases where $\la$ has jumps and/or where $\la$ does not have a product form.
 
 In order to simplify the exposition in this paper, we further assume the following set of conditions.
 \begin{Assumption}\label{ass:main-1} In \eqref{eq:jumps}, we have $\ga(s,z)= z_1$ for all $s\geq0$ and $z\in\R^{d'}$ and $F(dz)=dz$. Moreover,
in the notation of Assumption~\ref{ass:main}, the following holds under $\Q$ for all $n\in\N$, $t<T_n$, $t'\in[t,t+1]$ and $z\in\R^{d'}$:
 		\begin{equation}\label{eq:extra}  
 			\lvert z_1\rvert+\lvert\ga^\si(t,z)\rvert\leq J_n(z),\quad
\E[(\ga^\si(t',z)-\ga^\si(t,z))^2\wedge1]^{1/2}   \leq  \lvert t'-t\rvert^{1/2} J_n(z).
 		\end{equation} 
 \end{Assumption}
Under Assumption~\ref{ass:main-1}, price jumps are essentially like those of a time-changed Lévy process (e.g., if $\la(t,z)=\la_t f(z)$, then $x$ has the same jumps as $L_{\int_0^t \la_s ds}$, where $L$ is a Lévy process, independent of $\la$, with Lévy measure $f(z)dz$). In addition, both the jumps of price and volatility are of finite variation. While Assumption~\ref{ass:main-1} greatly simplifies the asymptotic expansions below, they are not strictly necessary; we refer the reader to Theorem~3.2 in \cite{CT23_a} for a more general version of the following result, which is Corollary 3.4 in \cite{CT23_a}:
\begin{theorem}\label{thm:main0}
	Suppose that $F$ is a $C^2$-function on $(0,\infty)$ and that Assumptions~\ref{ass:main} and \ref{ass:main-1} are satisfied. Then
	on the set $\{\inf_{s\in[\underline t,\overline t]}\si_s^2>0\}$, there is a finite number of It\^o semimartingales $v^{(k)}_t$ and $C^{(k)}_t(u)$, $k=1,\dots, K$,  such that $C^{(k)}_t(u)$ is uniformly bounded in $u$ on compacts and $\lvert \Delta^n_i v^{(k)}_t\rvert+\lvert \Delta^n_i C^{(k)}_t(u) \rvert= O^\uc (\sqrt{\Den})$, uniformly in $i$, and the following holds as $\Den\to0$, $T\to0$ and $\Den/T\to0$ for the approximations from \eqref{eq:incr-V}: 	for all $i$ with $i\Den=O(T)$, we have
	\begin{equation}\label{eq:incr-V-2} \begin{split}
			\Delta^{n}_i V_{t,T}(u)&=\Delta^n_i V_t+ \sum_{k=1}^K \Delta^n_i v^{(k)}_t C^{(k)}_t(u) T^n_{i-1} -T^n_{i-1}\si_{t^n_i}\int_{\R^{d'}} \Delta^n_i \ga^\si(t,z)\la(t^n_i,z)dz \\
			&\quad+O^\uc(T^{N/2}) + o^\uc(\sqrt{\Den}T+  \Den/\sqrt{T}),\\ 
			\Delta^{n}_i V_{t,T,T'}(u)&=\Delta^n_i V_t+O^\uc(T^{N/2}) + o^\uc(\sqrt{\Den}T+  \Den/\sqrt{T}).
		\end{split}
	\end{equation}
In particular, if $\inf_{s\in[\un t,\ov t]} \{(\si^\si_s)^2+(\ov\si^\si_s)^2\}>0$, $T^{N/2}= o^\uc(\sqrt{\Den}T)$ and $\Den/\sqrt{T} =O^\uc(\sqrt{\Den}T)$,   the bias of $	\Delta^{n}_i V_{t,T,T'}(u)$ is of strictly smaller asymptotic order than the bias of $	\Delta^{n}_i V_{t,T}(u)$.
\end{theorem}

\section{Proofs}\label{sec:proofs}

We shall deduce Theorems~\ref{thm:vov} and \ref{thm:lev} from a more general result concerning spot volatility estimation of an It\^o semimartingale observed with shrinking noise, which is of independent interest. To this end, let us look at two arbitrary  It\^o semimartingale processes
	\begin{align*}
		V_t &= V_0 + \int_0^t\alpha^V_sds + \int_0^t\sigma_s^V d\calw_s + \int_0^t\int_{\mathbb{R}}\gamma^V(s,z)\bone_{\{|\gamma^V(s,z)|\leq 1\}}(\un\pf-\un\qf)(ds,dz) \\&\quad+ \int_0^t\int_{\mathbb{R}}\gamma^V(s,z)\bone_{\{|\gamma^V(s,z)|> 1\}}\un\pf(ds,dz),\\
		X_t &= X_0 + \int_0^t\alpha^X_sds + \int_0^t\sigma_s^X d\calw_s +  \int_0^t\ov\sigma_s^X d\ov\calw_s +\int_0^t\int_{\mathbb{R}}\gamma^X(s,z)\bone_{\{|\gamma^X(s,z)|\leq 1\}}(\un\pf-\un\qf)(ds,dz) \\&\quad+ \int_0^t\int_{\mathbb{R}}\gamma^X(s,z)\bone_{\{|\gamma^X(s,z)|> 1\}}\un\pf(ds,dz),
	\end{align*}
which are assumed to satisfy the following regularity assumptions:

\bass \label{ass:Y}
Both   $V$ and $X$ are defined on $(\Om, \calf, \F,\P)$ with the following properties:
\begin{enumerate}
	\item The processes $\alpha^V$ and $\al^X$ are locally bounded.
	\item The processes $\si^V$, $\si^X$ and $\ov \si^X$ are c\`adl\`ag and $\calw$ and $\ov\calw$ are two independent standard $\F$-Brownian motions. Moreover, $\un\pf$ is an $\F$-Poisson random measure with intensity measure $\un \qf(ds,dz)=ds dz$.
	\item There is a localizing sequence $(T'_n)_{n\in\N}$) of stopping times and  a deterministic nonnegative function $J'_n(z)$ for each $n\in\N$  satisfying $\int_{\mathbb{R}} (J'_n(z))^{r'} \la(dz) < \infty$ with some $r'\in[0,\frac 43)$   such that $|\gamma^V(s,z)| \wedge 1 \leq J'_n(z)$ and $|\gamma^X(s,z)| \wedge 1 \leq J'_n(z)$  for all $(\omega, s, z)$ with $s \leq T'_n(\omega)$. 
\end{enumerate} 
\eass

We also introduce  observation errors $\epsilon_{i\Delta_n}^n$ for the process $V$, which are defined at the observation times $0,\Delta_n, 2\Delta_n,\ldots$ and satisfying the following assumptions:

\bass\label{ass:eps}
The observation errors $\epsilon_{i\Delta_n}^n$ of $V$ are defined on $(\ov \Om,\ov\calf,\ov\P)$ and $\calf$-conditionally independent   as $i$ varies. They further satisfy
\begin{equation}\label{eq:prop-eps}\begin{split}
		\mathbb{E}^{\ov \P} [\epsilon_{i\Delta_n}^n\mid\mathcal{F} ] \ &=0,\qquad \bigl\lvert\mathbb{E}^{\ov \P}  [(\epsilon_{i\Delta_n}^n)^2\mid \mathcal{F} ] - v_{i\Delta_n}\delta_n\bigr\rvert\leq C_{i\Delta_n}\delta_n\iota_n,\\
		\mathbb{E}^{\ov \P} [\lvert\epsilon_{i\Delta_n}^n\rvert^p\mid \mathcal{F}]&\leq K_pC_{i\Delta_n}\delta_n^{p/2}\qquad \text{for all } p>2,
\end{split}\end{equation}
for some deterministic sequences  $\delta_n$ and $\iota_n$  shrinking to zero as $n\rightarrow\infty$, some $\F$-adapted processes $v_t$  and $C_t$ with c\`adl\`ag paths, and some finite constant $K_p$. 	Moreover,  the limit
$\phi=\lim_{n\to\infty}\frac{\Delta_n}{\Delta_n+\delta_n}$
exists in $[0,1]$.
\eass

We denote   increments of  noisy observations of $V$ by 
\begin{equation}
	\Delta_i^n\widehat{V}_t  = \Delta_i^n V + \Delta^n_i \eps^n = V_{i\Den}-V_{(i-1)\Den}+ \epsilon^n_{i\Delta_n}-\epsilon^n_{(i-1)\Delta_n}.
\end{equation}
The process $X$ is assumed to be observed without errors.

\bprop\label{prop:auto}
Suppose that Assumption~\ref{ass:Y} holds for the processes $V$ and $X$ and Assumption~\ref{ass:eps} holds for the observation errors $\epsilon^n_{i\Delta_n}$. Consider a  deterministic integer sequence  $k_n$   satisfying $k_n\to\infty$  and
the truncation function $\tau_n(x) = x\bone_{\{\lvert x\rvert\leq \upsilon_n\}}$ with some $\upsilon_n>0$ such that
\begin{equation}\label{eq:cond_ell_n-1} 
k_n^2\Den\to0,\quad	
\Den=O(\upsilon_n^{2r'}),\quad (\Den+\delta_n)/\upsilon_n^2 = O(k_n^{-\iota})
\end{equation}
for some $\iota>0$ (and the same $r'$ as in Assumption~\ref{ass:Y}3). 
\begin{enumerate}
	\item The estimators 
	\begin{equation}\label{eq:auto_1} \begin{split}
			(\wh \si^{n,V}_0)^2=(\wh \si^{n,V}_0(\upsilon_n))^2 &=\frac{1}{k_n\Delta_n}\sum_{i=2}^{k_n}\Bigl((\tau_n(\Delta_i^n\widehat{V}))^2 + 2\tau_n(\Delta_{i-1}^n\widehat{V})\tau_n(\Delta_i^n\widehat{V})\Bigr),\\
			\wh \eta^n_0 = \wh \eta^n_0(\upsilon_n)& = \frac{1}{k_n\Delta_n}\sum_{i=1}^{k_n} \tau_n(\Delta_i^n\widehat{V})\tau_n(\Delta_i^n X)
		\end{split}
	\end{equation}
	satisfy the CLTs
	\begin{equation}\label{eq:auto_2} 	
		\frac{\sqrt{k_n}\Den}{\Delta_n+\delta_n} \Bigl( (\wh \si^{n,V}_0)^2- ( \si^{V}_0)^2\Bigr)  \xrightarrow{\mathcal{L}-s}\sqrt{\Sigma_0(\phi)+2\Sigma_1(\phi)}Z 
	\end{equation}
	and
	\begin{equation}\label{eq:auto-2-2}
		\frac{\sqrt{k_n}\Den}{\Delta_n+\delta_n} \Bigl(  \wh\eta^n_0 -\si_0^V\si_0^X\Bigr)  \xrightarrow{\mathcal{L}-s}\sqrt{\Sigma'(\phi)}Z', 
	\end{equation}
	where $Z$ and $Z'$ are standard normal variables, defined on a product extension of $(\ov \Om,\ov\calf,\ov\P)$  and independent from it, and
	\begin{equation}\label{eq:auto_3}\begin{split}
			\Sigma_0(\phi) &= 6(\si_0^V)^4\phi^2+16(\si_0^V)^2v_0 \phi(1-\phi)+22v_0^2(1-\phi)^2,\\
			\Sigma_1(\phi)& = -4(\si_0^V)^2v_0 \phi(1-\phi)+9v_0^2(1-\phi)^2,\\
			\Sigma'(\phi)&= (\si_0^V)^2[2(\si_0^X)^2+ (\ov\si_0^X)^2]\phi^2 + 2[(\si_0^X)^2+(\ov\si_0^X)^2]v_0\phi(1-\phi).
		\end{split}
	\end{equation}  
\item Defining $q^n_i(\wh V)=(\tau_n(\Delta_i^n\widehat{V}))^2 + 2\tau_n(\Delta_{i-1}^n\widehat{V})\tau_n(\Delta_i^n\widehat{V})$ and
\begin{equation}\label{eq:auto_4}
	\begin{split}
		\widehat{\Sigma}^n_0 &= \frac{1}{k_n(\Delta_n+\delta_n)^2}\Biggl(\sum_{i=2}^{k_n}\bigl(q^n_i(\wh V)\bigr)^2 - \sum_{i=4}^{k_n}q^n_{i-2}(\wh V)q^n_i(\wh V)\Biggr),\\ \widehat{\Sigma}^n_1 &= \frac{1}{k_n(\Delta_n+\delta_n)^2}\Biggl(\sum_{i=3}^{k_n}q^n_{i-1}(\wh V)q^n_i(\wh V) - \sum_{i=4}^{k_n}q^n_{i-2}(\wh V)q^n_i(\wh V)\Biggr)
	\end{split}
\end{equation}
and
\begin{equation}\label{eq:Sigma-prime}\begin{split}
		\widehat{\Sigma}^{\prime n} &= \frac{1}{k_n(\Delta_n+\delta_n)^2} \sum_{i=1}^{k_n} (\tau_n(\Delta^n_i \wh V))^2(\tau_n(\Delta^n_i X))^2\\
		&\quad  -\frac{1}{k_n(\Delta_n+\delta_n)^2} \sum_{i=3}^{k_n}\tau_n (\Delta^n_{i} X)\tau_n (\Delta^n_{i} \wh V)\tau_n (\Delta^n_{i-2} X)\tau_n (\Delta^n_{i-2} \wh V),
	\end{split}
\end{equation}
we have
\begin{equation}\label{eq:var-1}\begin{split}
		\widehat{\Sigma}^n_0-\frac{2}{k_n(\Den+\delta_n)^2}\sum_{i=1}^{k_n} \E[(\eps^n_{i\Den})^4] &\stackrel{\P}{\longrightarrow}\Sigma_0(\phi),\\
		\widehat{\Sigma}^n_1+\frac{1}{k_n(\Den+\delta_n)^2}\sum_{i=1}^{k_n} \E[(\eps^n_{i\Den})^4]&\stackrel{\P}{\longrightarrow}\Sigma_1(\phi)
	\end{split}
\end{equation}
and
\begin{equation}\label{eq:Sigma-prime-2}
	\wh \Sigma^{\prime n} \stackrel{\P}{\longrightarrow}\Sigma'(\phi).
\end{equation}
In particular,
\begin{equation}\label{eq:var-2} 
	\widehat{\Sigma}^n_0 + 2\widehat{\Sigma}^n_1 \stackrel{\P}{\longrightarrow} \Sigma_1(\phi)+2\Sigma_1(\phi).
\end{equation}
\item If $\phi=0$ in Assumption~\ref{ass:eps}, the limits in \eqref{eq:auto-2-2} and \eqref{eq:Sigma-prime-2} become degenerate. In this case, the nondegenerate limit theorems are given by
\begin{equation}\label{eq:auto-2-2-phi0}
	\sqrt{\frac{ k_n\Den}{\Delta_n+\delta_n}} \Bigl(  \wh\eta^n_0 -\si_0^V\si_0^X\Bigr)  \xrightarrow{\mathcal{L}-s}\sqrt{\Sigma''}Z', \qquad  \frac{\Den+\delta_n}{\Den}\wh \Sigma^{\prime n} \stackrel{\P}{\longrightarrow}\Sigma''. 
\end{equation}
where $\Sigma''=2[(\si_0^X)^2+(\ov\si^X_0)^2]v_0$.
\end{enumerate}

\eprop

	\begin{proof}  We first consider the case $\upsilon_n=\infty$ (i.e., no truncation). Then
	\begin{align*}
		&(\Delta_i^n\widehat{V})^2 + 2\Delta_{i-1}^n\widehat{V}\Delta_i^n\widehat{V}\\	&\quad=(\Delta_i^nV)^2 +2\Delta^n_{i-1}V\Delta^n_i V + 2(\Delta^n_i V\Delta^n_i\eps^n +\Delta^n_{i-1} \eps^n \Delta^n_i V+\Delta^n_{i-1}V\Delta^n_i\eps^n) \\
		&\quad\quad+ (\eps^n_{i\Den})^2 - (\eps^n_{(i-1)\Den})^2 + 4\eps^n_{(i-1)\Den}\eps^n_{i\Den} -2\eps^n_{(i-2)\Den}\eps^n_{i\Den} +2\eps^n_{(i-2)\Den}\eps^n_{(i-1)\Den}
	\end{align*}
and
\begin{equation*}
 \Delta^n_i \wh V \Delta^n_i X = \Delta^n_i V\Delta^n_i X + \Delta^n_i \eps^n \Delta^n_i X.
\end{equation*}
	By standard localization arguments, there is no loss of generality to assume that the processes $v$ and $C$ in Assumption~\ref{ass:eps} are bounded by a finite constant, say, $K$. 
	Then
	\begin{align*}
		\E\biggl[\biggl\lvert\frac{1}{\sqrt{k_n}(\Den+\delta_n)} \sum_{i=2}^{k_n}  \Bigl((\eps^n_{i\Den})^2 - (\eps^n_{(i-1)\Den})^2\Bigr) \biggr\rvert\biggr]&\leq \frac{\E[(\eps^n_{k_n\Den})^2+(\eps^n_{\Den})^2]}{\sqrt{k_n}(\Den+\delta_n)}\\
		&\leq  \frac{2K(1+\iota_n)\delta_n}{\sqrt{k_n}(\Den+\delta_n)} \to0. 
	\end{align*}
	Similarly, the expressions $\sum_{i=2}^{k_n} (\Delta^n_i V\Delta^n_i\eps^n +\Delta^n_{i-1} \eps^n \Delta^n_i V+\Delta^n_{i-1}V\Delta^n_i\eps^n)$ and $\sum_{i=2}^{k_n-2} (\Delta^n_{i-1} V - \Delta^n_{i+2}V)\eps^n_{i\Den}$ as well as $\sum_{i=1}^{k_n} \Delta^n_i \eps^n \Delta^n_i X$ and $ \sum_{i=2}^{k_n-2}(\Delta^n_i X-\Delta^n_{i+1} X)\eps^n_{i\Den}$  only differ by boundary terms that are negligible. Moreover, by Lemma 13.3.12 in  \cite{JP12}, we have $\frac{\sqrt{k_n}\Den}{\Delta_n+\delta_n}  ( \frac{1}{k_n\Den}\int_0^{k_n\Den}(\si_s^V)^2ds- ( \si^{V}_0)^2 ) \stackrel{\P}{\longrightarrow}0$ and 
	$\frac{\sqrt{k_n}\Den}{\Delta_n+\delta_n}  (  \frac{1}{k_n\Den}\int_0^{k_n\Den} \si_s^V\si_s^Xds-\si_0^V\si_0^X )\stackrel{\P}{\longrightarrow}0$.
	Therefore, in order to prove  \eqref{eq:auto_2} and \eqref{eq:auto-2-2}, it suffices to show that 
	\begin{align*}
		&\frac{1}{\sqrt{k_n}(\Den+\delta_n)}\sum_{i=2}^{k_n-2} \biggl((\Delta_i^n {V})^2-\int_{(i-1)\Den}^{i\Den} (\si_s^V)^2 ds + 2\Delta^n_{i-1}V\Delta^n_iV \\
		&\quad + 2(\Delta^n_{i-1} V - \Delta^n_{i+2}V)\eps^n_{i\Den}+ 6\eps^n_{(i-1)\Den}\eps^n_{i\Den} -2\eps^n_{(i-2)\Den}\eps^n_{i\Den} \biggr)
	\end{align*}
and 
\[ \frac{1}{\sqrt{k_n}(\Den+\delta_n)}\sum_{i=2}^{k_n-2}  \biggl(  \Delta^n_i V\Delta^n_i X - \int_{(i-1)\Den}^{i\Den}\si_s^V\si_s^X ds + (\Delta^n_i X-\Delta^n_{i+1} X)\eps^n_{i\Den} \biggr) \]
	  converge stably in law to the right-hand sides of \eqref{eq:auto_2} and \eqref{eq:auto-2-2}, respectively. But this follows from
	\begin{equation}\label{eq:conv1} 
		\frac{1}{\sqrt{k_n}(\Den+\delta_n)}\sum_{i=2}^{k_n-2} 	\begin{pmatrix} (\Delta_i^n {V})^2 -\int_{(i-1)\Den}^{i\Den} (\si_s^V)^2 ds\\ \Delta^n_i V\Delta^n_i X-\int_{(i-1)\Den}^{i\Den}\si_s^V\si_s^X ds\\  \Delta_{i-1}^n {V}\Delta_i^n{V} \\ (\Delta^n_{i-1} V- \Delta^n_{i+2}V)\eps^n_{i\Den} \\ (\Delta^n_{i} X- \Delta^n_{i+1}X)\eps^n_{i\Den}  \\ \eps^n_{(i-1)\Den}\eps^n_{i\Den}\\ \eps^n_{(i-2)\Den}\eps^n_{i\Den} \end{pmatrix}\xrightarrow{\mathcal{L}-s} \begin{pmatrix} \sqrt 2  (\si_0^V)^2 \phi Z_1\\ \sqrt 2 \si_0^V\si_0^X \phi Z_1 + \si_0^V\ov\si_0^X \phi Z_2 \\ (\si_0^V)^2 \phi Z_3\\ \sqrt{2(\si_0^V)^2v_0\phi(1-\phi)}Z_4 \\ \sqrt{2[(\si_0^X)^2+(\ov\si_0^X)^2]v_0\phi(1-\phi)}Z_5 \\  v_0 (1-\phi)Z_6 \\ v_0(1-\phi)Z_7\end{pmatrix}, \end{equation}
	where $Z_1,\dots,Z_7$ are independent standard normal random variables. In \eqref{eq:conv1}, the convergence of the first three coordinates can be shown by following the proof of Theorem~13.3.3  in \cite{JP12} and the proof of Theorem~8 in \cite{ALTZ21}. In particular, from the proofs we know that only the diffusive part of $V$ enters the limit, that is, we may assume that $\al^V\equiv \ga^V\equiv 0$. From here, we can then use Theorem~2.2.15 in \cite{JP12} to show \eqref{eq:conv1}. 
	
	To show \eqref{eq:auto_2} and \eqref{eq:auto-2-2} for general $\upsilon_n$,  we argue similarly to \cite{JP12}, Lemma~13.3.10: by localization, we may assume without loss of generality that the processes $\al^V$, $\si^V$ and $\int_\R \lvert (\lvert\ga^V(t,z)\rvert^{r'}\bone_{\{ \lvert\ga^V(t,z)\rvert\leq 1\}}+\lvert\ga^V(t,z)\rvert\bone_{\{ \lvert\ga^V(t,z)\rvert> 1\}})dz$ are uniformly bounded. Thus, writing $V'_t=V_t-V''_t$ and $V''_t=\int_0^t\int_\R \ga^V(s,z)(\un\pf-\un\qf)(ds,dz)$, we have a constant $L>0$ for which
	\begin{align*}
		&\ov \P\bigl (( \wh \si^{n,V}_0(\infty))^2 \neq (\wh\si^{n,V}_0(\upsilon_n))^2\bigr) \vee \ov \P\bigl (\wh \eta^{n}_0(\infty) \neq \wh\eta^{n}_0(\upsilon_n)\bigr)	\\
		&\qquad\leq \sum_{i=1}^{k_n} \ov\P(\lvert \Delta^n_i \wh V\rvert>\upsilon_n) \leq \sum_{i=1}^{k_n} \ov\P(\lvert \Delta^n_i   V'+\Delta^n_i\eps^n\rvert>\upsilon_n/2) + \sum_{i=1}^{k_n}  \P(\lvert \Delta^n_i   V''\rvert>\upsilon_n/2)  \\
		&\qquad\leq L^p K_p k_n\upsilon_n^{-p}(\Den+\delta_n)^{p/2} + k_n\Den\upsilon_n^{-{r'}},
	\end{align*}
	where in the last step $p>0$ is arbitrary and we applied Chebyshev's inequality, Assumption~\ref{ass:eps} and Equation~(13.2.23) in \cite{JP12}. Writing $k_n\Den\upsilon_n^{-{r'}} =(k_n^2\Den)^{1/2}(\Den/\upsilon_n^{2{r'}})^{1/2}  $, we see that the first two conditions in \eqref{eq:cond_ell_n-1} imply $k_n\Den\upsilon_n^{-{r'}}\to0$. Thanks to the last condition in \eqref{eq:cond_ell_n-1}, also the first term in the last line of the previous display goes to $0$ as soon as $p$ is chosen sufficiently large.

	If $\upsilon_n=\infty$,
	the statements concerning $\wh\Si^n_0$, $\wh \Si^n_1$ and $\wh\Si^{\prime n}$ can be shown via a    straightforward computation of the law of large number limits of the expressions in \eqref{eq:auto_4} and \eqref{eq:Sigma-prime}. For general $\upsilon_n$, everything remains the same because, as before, truncation has no asymptotic effect. Finally, the proof of \eqref{eq:auto-2-2-phi0} is completely analogous to that of \eqref{eq:auto-2-2} and \eqref{eq:Sigma-prime-2}.
\end{proof}

\begin{proof}[Proof of Theorem~\ref{thm:vov} and Theorem~\ref{thm:lev}] We only analyze $\wh{VV}^n_{t,T}(u)$; the proof for $\wh{VV}^n_{t,T,T'}(u)$, $\wh{LV}^n_{t,T}(u)$ and $\wh{LV}^n_{t,T,T'}(u)$ is similar and therefore omitted.
	On the set $\{\inf_{s\in[\underline t,\overline t]} \si_s^2 >0\}$, we have by a second-order Taylor expansion that
	\begin{align*}
		\Delta^n_i \wh V_{t,T}(u)	&=F\bigl(-\tfrac2{u^2} \log \lvert \wh\call_{t^n_{i-1},T^n_{i-1}}(u)\rvert\bigr)-F\bigl(-\tfrac2{u^2} \log \lvert \wh\call_{t^n_{i},T^n_{i}}(u)\rvert\bigr) \\
		&=\Delta^n_i V_{t,T}(u)+\eps^{n,i}_{t,T}(u)\\
		&\quad + O^\uc\bigl(\lvert \wh \call_{t^n_{i-1},T^n_{i-1}}(u)-\call_{t^n_{i-1},T^n_{i-1}}(u)\rvert^2+\lvert \wh \call_{t^n_{i},T^n_{i}}(u)-\call_{t^n_{i},T^n_{i}}(u)\rvert^2\bigr) ,
	\end{align*}
	where 
	\begin{align*}
		\eps^{n,i}_{t,T}(u)	&=-\frac{2}{u^2}F'\bigl(\si^2_{t^n_{i-1},T^n_{i-1}}(u)\bigr)\frac{\Re\bigl( \call_{t^n_{i-1},T^n_{i-1}}(u)\bigr)}{\lvert \call_{t^n_{i-1},T^n_{i-1}}(u)\rvert^2}\Re\bigl( \wh \call_{t^n_{i-1},T^n_{i-1}}(u)-\call_{t^n_{i-1},T^n_{i-1}}(u)\bigr) \\
		&\quad+\frac{2}{u^2}F'\bigl(\si^2_{t^n_{i},T^n_{i}}(u)\bigr)\frac{\Re\bigl( \call_{t^n_{i},T^n_{i}}(u)\bigr)}{\lvert \call_{t^n_{i},T^n_{i}}(u)\rvert^2}\Re\bigl( \wh \call_{t^n_{i},T^n_{i}}(u)-\call_{t^n_{i},T^n_{i}}(u)\bigr) \\
		&\quad-\frac{2}{u^2}F'\bigl(\si^2_{t^n_{i-1},T^n_{i-1}}(u)\bigr)\frac{\Im\bigl( \call_{t^n_{i-1},T^n_{i-1}}(u)\bigr)}{\lvert \call_{t^n_{i-1},T^n_{i-1}}(u)\rvert^2}\Im\bigl( \wh \call_{t^n_{i-1},T^n_{i-1}}(u)-\call_{t^n_{i-1},T^n_{i-1}}(u)\bigr) \\
		&\quad+\frac{2}{u^2}F'\bigl(\si^2_{t^n_{i},T^n_{i}}(u)\bigr)\frac{\Im\bigl( \call_{t^n_{i},T^n_{i}}(u)\bigr)}{\lvert \call_{t^n_{i},T^n_{i}}(u)\rvert^2}\Im\bigl( \wh \call_{t^n_{i},T^n_{i}}(u)-\call_{t^n_{i},T^n_{i}}(u)\bigr). 
	\end{align*}
	By Theorem~2 in \cite{T21}, we know that
	\[ O^\uc\bigl(\lvert \wh \call_{t^n_{i-1},T^n_{i-1}}(u)-\call_{t^n_{i-1},T^n_{i-1}}(u)\rvert^2+\lvert \wh \call_{t^n_{i},T^n_{i}}(u)-\call_{t^n_{i},T^n_{i}}(u)\rvert^2\bigr) =O^\uc (\delta/\sqrt{T}).
	\]
	So  in conjunction with Theorem~\ref{thm:main0} and  Assumption~\ref{ass:hf}, we obtain
	\begin{equation}\label{eq:Delta-Vhat} \begin{split}
			\Delta^n_i \wh V_{t,T}(u)&=\Delta^n_i V_{t,T}(u)+\eps^{n,i}_{t,T}(u)+O^\uc(\delta/\sqrt{T}  )\\
&=\Delta^n_i V_t + \eps^{n,i}_{t,T}(u)+O^\uc(\delta/\sqrt{T})+O^\uc(T^{N/2})+O^\uc(\sqrt{\Den}T)+o^\uc(\Den/\sqrt{T})\\
			&=\Delta^n_i V_t + \eps^{n,i}_{t,T}(u)+O^\uc(\delta/\sqrt{T})+o^\uc(\sqrt{\Den/k_n}).
		\end{split}
	\end{equation}
(If we expand $\Delta^n_i \wh V_{t,T,T'}(u)$ in a similar fashion, the $O^\uc(\sqrt{\Den}T)$-term becomes $o^\uc(\sqrt{\Den}T)$, which is why it suffices to assume $k_nT=O(1)$ instead of $k_nT\to0$ for the convergence of $\wh{VV}^n_{t,T,T'}(u)$ in \eqref{eq:main-conv-alt} as well as the convergence of  
$\wh{LV}^n_{t,T,T'}(u)$ in    \eqref{eq:main-conv-22}.)
	
	Next, consider the error variables $\eps^{n,i}_{t,T}(u)$ more closely. 	By \eqref{eq:L_hat}, we have 
	\begin{equation}
		\widehat{\mathcal{L}}_{s,T}(u) - \mathcal{L}_{s,T}(u) = Z_{s,T}(u)+\mathcal{R}_{s,T}(u) 
	\end{equation}
	for any $s\in[\underline{t},\overline t]$, where
	\begin{equation}\label{eq:Z}\begin{split}
			Z_{s,T}(u) = -\left(\frac{u^2}{T}+i\frac{u}{\sqrt{T}}\right)e^{-x_s}\sum_{j=2}^{N_{s,T}}e^{(iu/\sqrt{T}-1)(k_{j-1,s,T}-x_s)}\epsilon_{j-1,s,T}\delta_{j,s,T}
	\end{split}\end{equation}
	and $\mathcal{R}_{s,T}(u) =  (\frac{u^2}{T}+i\frac{u}{\sqrt{T}} )e^{-x_s}\overline{\mathcal{R}}_{s,T}(u)$ with
	\begin{align*}
		&\overline{\mathcal{R}}_{s,T}(u) = \left(\int_{-\infty}^{\underline{k}_{s,T}}e^{(iu/\sqrt{T}-1)(k-x_s)}O_{s,T}(k)dk+\int_{\overline{k}_{s,T}}^{\infty}e^{(iu/\sqrt{T}-1)(k-x_s)}O_{s,T}(k)dk\right)\\
		&\quad -\sum_{j=2}^{N_{s,T}}\int_{k_{j-1,s,T}}^{k_{j,s,T}}\left(e^{(iu/\sqrt{T}-1)(k_{j-1,s,T}-x_s)}O_{s,T}(k_{j-1,s,T})-e^{(iu/\sqrt{T}-1)(k-x_s)}O_{s,T}(k)\right)dk.
	\end{align*}
	Using (\ref{bounds_1}) and (\ref{bounds_2}) from Lemma~\ref{lemma:bounds} below together with Assumption~\ref{ass:C}3, we can  bound 
	\begin{equation}
		\left|\mathcal{R}_{s,T}(u)\right| \leq C_s(u)\left(e^{2\underline{k}_{s,T}}+e^{-2\overline{k}_{s,T}} + \frac{\delta}{\sqrt{T}}\log T\right)=O^\uc\biggl(\frac{\delta}{\sqrt{T}}\log T\biggr).
	\end{equation}
	This shows that 
	\begin{equation}\label{eq:Delta-Vhat-2} \begin{split}
			\Delta^n_i \wh V_{t,T}(u)
			&=\Delta^n_i V_t + \ov\eps^{n,i}_{t,T}(u)+O^\uc\biggl(\frac{\delta}{\sqrt{T} }\log T +\Den\biggr)+o^\uc(\sqrt{\Den/k_n}),
		\end{split}
	\end{equation}
	where
	\begin{equation}\label{eq:ov-eps} 
		\begin{split}
			\ov	\eps^{n,i}_{t,T}(u)	&=-\frac{2}{u^2}F'\bigl(\si^2_{t^n_{i-1},T^n_{i-1}}(u)\bigr)\\
			&\quad\quad\times\frac{\Re\bigl( \call_{t^n_{i-1},T^n_{i-1}}(u)\bigr)\Re \bigl(Z_{t^n_{i-1},T^n_{i-1}}(u)\bigr)+\Im\bigl( \call_{t^n_{i-1},T^n_{i-1}}(u)\bigr)\Im \bigl( Z_{t^n_{i-1},T^n_{i-1}}(u)\bigr)}{\lvert \call_{t^n_{i-1},T^n_{i-1}}(u)\rvert^2} \\
			&\quad+\frac{2}{u^2}F'\bigl(\si^2_{t^n_{i},T^n_{i}}(u)\bigr)\frac{\Re\bigl( \call_{t^n_{i},T^n_{i}}(u)\bigr)\Re \bigl(Z_{t^n_{i},T^n_{i}}(u)\bigr)+\Im\bigl( \call_{t^n_{i},T^n_{i}}(u)\bigr)\Im \bigl(Z_{t^n_{i},T^n_{i}}(u)\bigr)}{\lvert \call_{t^n_{i},T^n_{i}}(u)\rvert^2}.
		\end{split}\raisetag{-3\baselineskip}
	\end{equation}
	
	Suppose that $\upsilon_n=\infty$ (i.e., no truncation). It is then an easy consequence of the assumptions $k_n(\delta/\sqrt{T})\log T\to0$ and $k_n\Den\to0$  that the terms within $O^\uc((\delta/\sqrt{T} )\log T +\Den)+o^\uc(\sqrt{\Den/k_n})$ are asymptotic negligible for the  convergence of $\wh{VV}^n_{t,T}(u)$ in \eqref{eq:main-conv-alt}. For general $\upsilon_n$, as long as \eqref{eq:cond_ell_n-1} is satisfied, one can argue as in the proof of Proposition~\ref{prop:auto} to show that truncation has no asymptotic impact. Therefore, the  convergence of $\wh{VV}^n_{t,T}(u)$ in \eqref{eq:main-conv-alt} follows from Proposition~\ref{prop:auto} (applied to $V$ and with  $\delta_n=\delta/\sqrt{T}$), provided we can verify the conditions listed in Assumption~\ref{ass:eps} for $\ov\eps^{n,i}_{t,T}(u)$. The $\calf$-conditional independence of $\ov\eps^{n,i}_{t,T}(u)$ and the fact that $\E^{\ov \P}[\ov\eps^{n,i}_{t,T}(u)]=0$ follow from \eqref{eq:Z}, \eqref{eq:ov-eps} and Assumption~\ref{ass:C}4. For the other two properties in \eqref{eq:prop-eps}, we may as usual assume that the process $C_t$ in Lemma~\ref{lemma:bounds} and the coefficients of $x_t$ are all uniformly bounded. Thus, if 
	
\begin{equation}\label{eq:Psi_tilde}
\wt \Phi(k)=\vp(k)+\lvert k\rvert\Phi(-\lvert k\rvert),
\end{equation}
where $\vp$ and $\Phi$ are the standard normal density and distribution function, respectively, the proof of Theorem~3 in \cite{T19} shows that one can find a nonnegative c\`adl\`ag process $K_t$ and finite constants $K(p)$ such that 
\begin{equation}\label{eq:aux}
	\begin{split}
		&\left|\frac{\sqrt{T}}{\delta}\mathbb{E}[\Re(Z_{t_i^n,T_i^n}(u))^2\mid \mathcal{F}] - u^4\sigma_{t_i^n}^3\zeta_{t_i^n}^2(0)\int_{\mathbb{R}}\sin^2(\sigma_{t_i^n}uk)\widetilde{\Phi}^2(k)dk\right|\leq K_{t^n_i} \iota_n,\\&
		\left|\frac{\sqrt{T}}{\delta}\mathbb{E}[\Im(Z_{t_i^n,T_i^n}(u))^2\mid \mathcal{F}] - u^4\sigma_{t_i^n}^3\zeta_{t_i^n}^2(0)\int_{\mathbb{R}}\cos^2(\sigma_{t_i^n}uk)\widetilde{\Phi}^2(k)dk\right|\leq K_{t^n_i}  \iota_n,\\&
		\left|\frac{\sqrt{T}}{\delta}\mathbb{E}[\Re(Z_{t_i^n,T_i^n}(u))\Im(Z_{t_i^n,T_i^n}(u))\mid \mathcal{F}]\right|\leq K_{t^n_i}  \iota_n,\\
		&\mathbb{E}[|Z_{t_i^n,T_i^n}(u)|^{p}\mid \mathcal{F}]\leq K(p)K^p_{t^n_i} \left(\frac{\sqrt{\delta}}{T^{1/4}}\right)^{p},\qquad p>2,
	\end{split}
\end{equation}
locally uniformly in $u$, with
a deterministic sequence $\iota_n\downarrow0$.
\end{proof}

For the proof of Theorem~\ref{thm:vov}, we used the following lemma,  which is Lemma~3 in \cite{T21}:
\begin{lemma}\label{lemma:bounds}
Suppose that Assumptions \ref{ass:main} and \ref{ass:C}1 hold. For any $t>0$, there is a process  $C$ with c\`{a}dl\`{a}g paths   such that on the set $\{\inf_{s\in[\underline t,\overline t]} \si_s^2>0\}$ we have 
\begin{equation}\label{bounds_1}
	O_{s,T}(k)\leq C_s\biggl(Te^{3k}\bone_{\{k-x_s<-1\}}+Te^{-k}\bone_{\{k-x_s>1\}}+\biggl(\sqrt{T}\wedge\frac{T}{|k-x_s|}\biggr)\bone_{\{|k-x_s|<1\}}\biggr)
\end{equation}
and
\begin{equation}\label{bounds_2}
	|O_{s,T}(k_1)-O_{s,T}(k_2)|\leq C_s\left[\frac{T}{(k_2-x_s)^4}\wedge\frac{T}{(k_2-x_s)^2}\wedge 1\right]|e^{k_1}-e^{k_2}|
\end{equation}
for $s\in[\underline{t},\overline{t}]$, $k\in\R$, $k_1<k_2<x_s$ or $k_1>k_2>x_s$, and sufficiently small $T$.
\end{lemma}

\bibliographystyle{abbrvnat}
\bibliography{ovv}

\begin{thebibliography}{43}
\providecommand{\natexlab}[1]{#1}
\providecommand{\url}[1]{\texttt{#1}}
\expandafter\ifx\csname urlstyle\endcsname\relax
  \providecommand{\doi}[1]{doi: #1}\else
  \providecommand{\doi}{doi: \begingroup \urlstyle{rm}\Url}\fi

\bibitem[Agarwal et~al.(2017)Agarwal, Arisoy, and Naik]{agarwal2017volatility}
V.~Agarwal, Y.~E. Arisoy, and N.~Y. Naik.
\newblock Volatility of aggregate volatility and hedge fund returns.
\newblock \emph{Journal of Financial Economics}, 125\penalty0 (3):\penalty0
  491--510, 2017.

\bibitem[A{\"\i}t-Sahalia and Jacod(2014)]{ait2014high}
Y.~A{\"\i}t-Sahalia and J.~Jacod.
\newblock \emph{High-Frequency Financial Econometrics}.
\newblock Princeton University Press, 2014.

\bibitem[A\"{\i}t-Sahalia et~al.(2017)A\"{\i}t-Sahalia, Fan, Laeven, Wang, and
  Yang]{AFLWY17}
Y.~A\"{\i}t-Sahalia, J.~Fan, R.~J.~A. Laeven, C.~D. Wang, and X.~Yang.
\newblock Estimation of the continuous and discontinuous leverage effects.
\newblock \emph{Journal of the American Statistical Association}, 112\penalty0
  (520):\penalty0 1744--1758, 2017.

\bibitem[Andersen et~al.(2015{\natexlab{a}})Andersen, Bondarenko, and
  Gonzalez-Perez]{andersen2015exploring}
T.~G. Andersen, O.~Bondarenko, and M.~T. Gonzalez-Perez.
\newblock Exploring return dynamics via corridor implied volatility.
\newblock \emph{The Review of Financial Studies}, 28\penalty0 (10):\penalty0
  2902--2945, 2015{\natexlab{a}}.

\bibitem[Andersen et~al.(2015{\natexlab{b}})Andersen, Fusari, and
  {Todorov}]{AFT}
T.~G. Andersen, N.~Fusari, and V.~{Todorov}.
\newblock Parametric inference and dynamic state recovery from option panels.
\newblock \emph{Econometrica}, 83:\penalty0 1081--1145, 2015{\natexlab{b}}.

\bibitem[Andersen et~al.(2015{\natexlab{c}})Andersen, Fusari, and
  {Todorov}]{AFT_b}
T.~G. Andersen, N.~Fusari, and V.~{Todorov}.
\newblock The risk premia embedded in index options.
\newblock \emph{Journal of Financial Economics}, 117:\penalty0 558--584,
  2015{\natexlab{c}}.

\bibitem[Andersen et~al.(2023)Andersen, Li, Todorov, and Zhou]{ALTZ21}
T.~G. Andersen, Y.~Li, V.~Todorov, and B.~Zhou.
\newblock Volatility measurement with pockets of extreme return persistence.
\newblock \emph{Journal of Econometrics}, 237\penalty0 (2):\penalty0 105048,
  2023.

\bibitem[Bakshi and Madan(2000)]{bakshi2000spanning}
G.~Bakshi and D.~Madan.
\newblock Spanning and derivative-security valuation.
\newblock \emph{Journal of Financial Economics}, 55\penalty0 (2):\penalty0
  205--238, 2000.

\bibitem[Bandi et~al.(2023{\natexlab{a}})Bandi, Fusari, and Ren{\`o}]{BFR23}
F.~M. Bandi, N.~Fusari, and R.~Ren{\`o}.
\newblock Structural stochastic volatility.
\newblock \emph{SSRN preprint}, 2023{\natexlab{a}}.

\bibitem[Bandi et~al.(2023{\natexlab{b}})Bandi, Fusari, and Ren{\`o}]{BFR23b}
F.~M. Bandi, N.~Fusari, and R.~Ren{\`o}.
\newblock 0{DTE} option pricing.
\newblock \emph{SSRN preprint}, 2023{\natexlab{b}}.

\bibitem[Bekaert and Wu(2000)]{bekaert2000asymmetric}
G.~Bekaert and G.~Wu.
\newblock Asymmetric volatility and risk in equity markets.
\newblock \emph{The Review of Financial Studies}, 13\penalty0 (1):\penalty0
  1--42, 2000.

\bibitem[Black(1976)]{black1976studies}
F.~Black.
\newblock Studies of stock market volatility changes.
\newblock \emph{Proceedings of the American Statistical Association, Business
  \& Economic Statistics Section, 1976}, 1976.

\bibitem[Bollerslev and Todorov(2023)]{bollerslev2023jump}
T.~Bollerslev and V.~Todorov.
\newblock The jump leverage risk premium.
\newblock \emph{Journal of Financial Economics}, 150\penalty0 (3):\penalty0
  103723, 2023.

\bibitem[Bollerslev et~al.(2012)Bollerslev, Sizova, and
  Tauchen]{bollerslev2012volatility}
T.~Bollerslev, N.~Sizova, and G.~Tauchen.
\newblock Volatility in equilibrium: Asymmetries and dynamic dependencies.
\newblock \emph{Review of Finance}, 16\penalty0 (1):\penalty0 31--80, 2012.

\bibitem[Campbell and Hentschel(1992)]{campbell1992no}
J.~Y. Campbell and L.~Hentschel.
\newblock No news is good news: An asymmetric model of changing volatility in
  stock returns.
\newblock \emph{Journal of Financial Economics}, 31\penalty0 (3):\penalty0
  281--318, 1992.

\bibitem[Carr and Madan(2001)]{CM01}
P.~Carr and D.~Madan.
\newblock {Optimal} {Positioning} in {Derivative} {Securities}.
\newblock \emph{Quantitative Finance}, 1:\penalty0 19--37, 2001.

\bibitem[Chen et~al.(2022)Chen, Chordia, Chung, and Lin]{chen2022volatility}
T.-F. Chen, T.~Chordia, S.-L. Chung, and J.-C. Lin.
\newblock Volatility-of-volatility risk in asset pricing.
\newblock \emph{The Review of Asset Pricing Studies}, 12\penalty0 (1):\penalty0
  289--335, 2022.

\bibitem[Chernov et~al.(2003)Chernov, Gallant, Ghysels, and
  Tauchen]{chernov2003alternative}
M.~Chernov, A.~R. Gallant, E.~Ghysels, and G.~Tauchen.
\newblock Alternative models for stock price dynamics.
\newblock \emph{Journal of Econometrics}, 116\penalty0 (1-2):\penalty0
  225--257, 2003.

\bibitem[Chong and Todorov(2023)]{CT23_a}
C.~H. Chong and V.~Todorov.
\newblock Asymptotic expansions for high-frequency option data.
\newblock \emph{SSRN preprint}, 2023.

\bibitem[Christie(1982)]{christie1982stochastic}
A.~A. Christie.
\newblock The stochastic behavior of common stock variances: Value, leverage
  and interest rate effects.
\newblock \emph{Journal of Financial Economics}, 10\penalty0 (4):\penalty0
  407--432, 1982.

\bibitem[Clinet and Potiron(2021)]{clinet2021estimation}
S.~Clinet and Y.~Potiron.
\newblock Estimation for high-frequency data under parametric market
  microstructure noise.
\newblock \emph{Annals of the Institute of Statistical Mathematics},
  73\penalty0 (4):\penalty0 649--669, 2021.

\bibitem[Duffie et~al.(2000)Duffie, Pan, and Singleton]{DPS00}
D.~Duffie, J.~Pan, and K.~Singleton.
\newblock {Transform} {Analysis} and {Asset} {Pricing} for {Affine}
  {Jump}-{Diffusions}.
\newblock \emph{Econometrica}, 68:\penalty0 1343--1376, 2000.

\bibitem[Engle and Ng(1993)]{engle1993measuring}
R.~F. Engle and V.~K. Ng.
\newblock Measuring and testing the impact of news on volatility.
\newblock \emph{The Journal of Finance}, 48\penalty0 (5):\penalty0 1749--1778,
  1993.

\bibitem[French et~al.(1987)French, Schwert, and Stambaugh]{french1987expected}
K.~R. French, G.~W. Schwert, and R.~F. Stambaugh.
\newblock Expected stock returns and volatility.
\newblock \emph{Journal of Financial Economics}, 19\penalty0 (1):\penalty0
  3--29, 1987.

\bibitem[Heston(1993)]{heston}
S.~L. Heston.
\newblock A closed-form solution for options with stochastic volatility with
  applications to bond and currency options.
\newblock \emph{The Review of Financial Studies}, 6\penalty0 (2):\penalty0
  327--343, 1993.

\bibitem[Hollstein and Prokopczuk(2018)]{hollstein2018aggregate}
F.~Hollstein and M.~Prokopczuk.
\newblock How aggregate volatility-of-volatility affects stock returns.
\newblock \emph{The Review of Asset Pricing Studies}, 8\penalty0 (2):\penalty0
  253--292, 2018.

\bibitem[Huang et~al.(2019)Huang, Schlag, Shaliastovich, and
  Thimme]{huang2019volatility}
D.~Huang, C.~Schlag, I.~Shaliastovich, and J.~Thimme.
\newblock Volatility-of-volatility risk.
\newblock \emph{Journal of Financial and Quantitative Analysis}, 54\penalty0
  (6):\penalty0 2423--2452, 2019.

\bibitem[Jacod and Protter(2012)]{JP12}
J.~Jacod and P.~Protter.
\newblock \emph{Discretization of Processes}, volume~67 of \emph{Stochastic
  Modelling and Applied Probability}.
\newblock Springer, Heidelberg, 2012.

\bibitem[Jacod and Shiryaev(2003)]{JS03}
J.~Jacod and A.~N. Shiryaev.
\newblock \emph{Limit Theorems for Stochastic Processes}, volume 288 of
  \emph{Grundlehren der mathematischen Wissenschaften [Fundamental Principles
  of Mathematical Sciences]}.
\newblock Springer-Verlag, Berlin, second edition, 2003.

\bibitem[Jacod and Todorov(2014)]{jacod2014efficient}
J.~Jacod and V.~Todorov.
\newblock Efficient estimation of integrated volatility in presence of infinite
  variation jumps.
\newblock \emph{The Annals of Statistics}, 42\penalty0 (3):\penalty0
  1029--1069, 2014.

\bibitem[Kalnina and Xiu(2017)]{KX17}
I.~Kalnina and D.~Xiu.
\newblock Nonparametric estimation of the leverage effect: A trade-off between
  robustness and efficiency.
\newblock \emph{Journal of the American Statistical Association}, 112\penalty0
  (517):\penalty0 384--396, 2017.

\bibitem[Li et~al.(2022)Li, Liu, and Zhang]{li2022volatility}
Y.~Li, G.~Liu, and Z.~Zhang.
\newblock Volatility of volatility: Estimation and tests based on noisy high
  frequency data with jumps.
\newblock \emph{Journal of Econometrics}, 229\penalty0 (2):\penalty0 422--451,
  2022.

\bibitem[Liu et~al.(2018)Liu, Liu, and Liu]{LLL18}
Q.~Liu, Y.~Liu, and Z.~Liu.
\newblock Estimating spot volatility in the presence of infinite variation
  jumps.
\newblock \emph{Stochastic Processes and their Applications}, 128\penalty0
  (6):\penalty0 1958--1987, 2018.

\bibitem[Sanfelici et~al.(2015)Sanfelici, Curato, and
  Mancino]{sanfelici2015high}
S.~Sanfelici, I.~V. Curato, and M.~E. Mancino.
\newblock High-frequency volatility of volatility estimation free from spot
  volatility estimates.
\newblock \emph{Quantitative Finance}, 15\penalty0 (8):\penalty0 1331--1345,
  2015.

\bibitem[Todorov(2019)]{T19}
V.~Todorov.
\newblock Nonparametric spot volatility from options.
\newblock \emph{The Annals of Applied Probability}, 29\penalty0 (6):\penalty0
  3590--3636, 2019.

\bibitem[Todorov(2021)]{T21}
V.~Todorov.
\newblock Higher-order small time asymptotic expansion of {I}t{\^o}
  semimartingale characteristic function with application to estimation of
  leverage from options.
\newblock \emph{Stochastic Processes and their Applications}, 142:\penalty0
  671--705, 2021.

\bibitem[Todorov and Zhang(2022)]{todorov2022information}
V.~Todorov and Y.~Zhang.
\newblock Information gains from using short-dated options for measuring and
  forecasting volatility.
\newblock \emph{Journal of Applied Econometrics}, 37\penalty0 (2):\penalty0
  368--391, 2022.

\bibitem[Todorov and Zhang(2023)]{todorov2021bias}
V.~Todorov and Y.~Zhang.
\newblock Bias reduction in spot volatility estimation from options.
\newblock \emph{Journal of Econometrics}, 234\penalty0 (1):\penalty0 53--81,
  2023.

\bibitem[Toscano et~al.(2022)Toscano, Livieri, Mancino, and
  Marmi]{toscano2022volatility}
G.~Toscano, G.~Livieri, M.~E. Mancino, and S.~Marmi.
\newblock Volatility of volatility estimation: central limit theorems for the
  fourier transform estimator and empirical study of the daily time series
  stylized facts.
\newblock \emph{Journal of Financial Econometrics}, 2022.

\bibitem[Vetter(2015)]{vetter2015estimation}
M.~Vetter.
\newblock Estimation of integrated volatility of volatility with applications
  to goodness-of-fit testing.
\newblock \emph{Bernoulli}, 21\penalty0 (4):\penalty0 2393--2418, 2015.

\bibitem[Wang and Mykland(2014)]{WM14}
C.~D. Wang and P.~A. Mykland.
\newblock The estimation of leverage effect with high-frequency data.
\newblock \emph{Journal of the American Statistical Association}, 109\penalty0
  (505):\penalty0 197--215, 2014.

\bibitem[Yang(2023)]{yang2023estimation}
X.~Yang.
\newblock Estimation of leverage effect: Kernel function and efficiency.
\newblock \emph{Journal of Business \& Economic Statistics}, 41\penalty0
  (3):\penalty0 939--956, 2023.

\bibitem[Yu(2005)]{yu2005leverage}
J.~Yu.
\newblock On leverage in a stochastic volatility model.
\newblock \emph{Journal of Econometrics}, 127\penalty0 (2):\penalty0 165--178,
  2005.

\end{thebibliography}

\end{document}